\documentclass[a4paper,11pt,twoside]{article}
\pdfoutput=1

\usepackage{geometry} 
\usepackage{amsmath}
\usepackage{amssymb}
\usepackage{color}
\usepackage{graphicx}
\usepackage[]{units} 
\usepackage[nosort]{cite}

\usepackage{hyperref}
\usepackage{color}
\definecolor{darkblue}{cmyk}{1,0.5,0,0.2}
\hypersetup{colorlinks, bookmarksnumbered, citecolor=darkblue, linkcolor=darkblue, pdfstartview=FitH, urlcolor=darkblue, linktocpage}
\newcommand{\arXhref}[1]{\href{http://arxiv.org/abs/#1}{\tt#1}} 

\DeclareMathOperator{\re}{Re}

\DeclareMathOperator{\diag}{Diag}

\newcommand{\fracwithdelims}[4]{\left#1 \frac{#3}{#4} \right#2}
\newcommand{\ord}[1]{\mathcal{O}\left( #1 \right)}

\newcommand{\Fig}[1]{Fig.~\ref{fig:#1}}

\newcommand{\Figs}[1]{Figs.~\ref{fig:#1}}

\newcommand{\Eq}[1]{Eq.~(\ref{eq:#1})}
\newcommand{\eq}[1]{eq.~(\ref{eq:#1})}
\newcommand{\Eqs}[1]{Eqs.~(\ref{eq:#1})}
\newcommand{\eqs}[1]{eqs.~(\ref{eq:#1})}
\newcommand{\U}{\text{U}}

\makeatletter
\renewcommand{\section}{\@startsection{section}{1}{0em}%
        {-3.5ex \@plus -1ex \@minus -.2ex}%
        {2.3ex \@plus.2ex}%
        {\normalfont\large\bfseries}}
\renewcommand{\subsection}{\@startsection{subsection}{2}{0em}%
        {-3.25ex\@plus -1ex \@minus -.2ex}%
        {1.5ex \@plus .2ex}%
        {\normalfont\bfseries}}
\renewcommand{\subsubsection}%
        {\@startsection{subsubsection}{3}{0em}%
        {-3.25ex\@plus -1ex \@minus -.2ex}%
        {1.5ex \@plus .2ex}%
        {\normalfont\itshape}}
\makeatother

\newlength{\myem}
\newcommand{\sep}[1]{#1}
\newcounter{mysubequation}[equation]
\renewcommand{\themysubequation}{\alph{mysubequation}}
\newcommand{\mytag}{\stepcounter{mysubequation}%
\tag{\theequation\protect\sep{\themysubequation}}}
\newcommand{\globallabel}[1]{\refstepcounter{equation}\label{#1}}

\renewcommand{\det}{\text{det}\,}

\usepackage{amsthm}
\theoremstyle{definition}
\newtheorem{theorem}{Theorem}

\newtheorem*{lemma*}{Lemma}
\newtheorem{proposition}[theorem]{Proposition}
\newtheorem*{proposition*}{Proposition}

\theoremstyle{definition}
\newtheorem*{definition}{Definition}
\newtheorem{example}{Example}

\theoremstyle{remark}

\newcommand{\be}{\begin{equation}}
\newcommand{\ee}{\end{equation}}

\newcommand{\SISSA}{SISSA/ISAS and INFN, I--34136 Trieste, Italy}
\newcommand{\ICTP}{ICTP, Strada Costiera 11, I--34151 Trieste, Italy}

\newcommand{\preprintnumber}{%
SISSA  44/2014/FISI}
 
\newcommand{\titletext}{Stable fermion mass matrices \\ and the charged lepton contribution to neutrino mixing} 
\newcommand{\authortext}{\large David Marzocca$^{\, a}$ and Andrea Romanino$^{\, a, b}$
\medskip\\\em\normalsize 
$\mbox{}^a$\SISSA \\[0.1\baselineskip] 
$\mbox{}^b$\ICTP
}
\newcommand{\abstracttext}{We study the general properties of hierarchical fermion mass matrices in which the small eigenvalues are stable with respect to perturbations of the matrix entries and we consider specific applications to the charged lepton contribution to neutrino mixing. In particular, we show that the latter can account for the whole lepton mixing. In this case a value of $\sin \theta_{13} \gtrsim m_e/m_\mu \sin\theta_{23}  \approx 0.03$, as observed, can be obtained without the need of any fine-tuning, and present data allow to determine the last row of the charged lepton mass matrix with good accuracy. We also consider the case in which the neutrino sector only provides a maximal 12 rotation and show that i) present data provide a $2\sigma$ evidence for a non-vanishing $31$ entry of the charged lepton mass matrix and ii) a plausible texture for the latter can account at the same time for the atmospheric mixing angle, the $\theta_{13}$ angle, and the deviation of the $\theta_{12}$ angle from $\pi/2$ without fine-tuning or tension with data. Finally, we show that the so-called ``inverted order'' of the 12 and 23 rotations in the charged lepton sector can  be obtained without fine-tuning, up to corrections of order $m_e/m_\mu$. }

\title{
\normalsize
\begin{tabular}[t]{l}
\end{tabular}
\hspace*{\fill}
\begin{tabular}[t]{l}\preprintnumber\end{tabular}
\vspace{3\baselineskip}\\\Large\bfseries\titletext\bigskip}
\author{\begin{minipage}[t]{0.8\textwidth}
\normalsize\centering\authortext
\end{minipage}}
\date{}

\begin{document}

\bigskip
\maketitle
\begin{abstract}\vspace*{2mm}\normalsize\noindent
\abstracttext
\end{abstract}\normalsize\vspace{\baselineskip}



\section{Introduction}

The experimental determination of lepton mass and mixing parameters has made remarkable progress in the last 15 years, gradually unveiling an unexpected pattern, which has often challenged the theoretical prejudice. Such an experimental information is essential to the ambitious program of understanding the origin of flavour breaking. This program has been most often carried out in a top-down approach based on flavour symmetries or other organizing principles. In this paper we would like to revisit the problem from a different point of view, in a bottom-up approach based on a general ``stability'' assumption, according to which the smallness of some fermion masses does not arise from special correlations among the entries of the mass matrix, and as a consequence it is stable with respect to small variations of the matrix entries. 

Our analysis will lead to constraints on the structure of fermion mass matrices. The latter contain of course additional parameters that are not physical in the Standard Model (SM) --- their form depends in particular on the basis in flavour space in which they are written. The idea underlying our approach is that in a certain basis in flavour space, associated to the unknown physics from which they originate, the entries of the fermion mass matrix can be considered as \emph{independent} fundamental parameters, i.e.\ parameters that are not correlated, neither as a consequence of a non-abelian symmetry, nor accidentally. We consider such an assumption motivated and timely, as an experimental evidence of such correlations, which would have been welcome as a smoking gun of underlying symmetries, failed so far to show up in the measurement of $\theta_{13}$ and $\theta_{23}$~\cite{Capozzi:2013csa,GonzalezGarcia:2012sz}. For example, neutrino mass models leading to the so-called ``tri-bimaximal'' (TBM) mixing structure \cite{TBM} for the neutrino mass matrix $m_\nu$ require 3 independent correlations among the entries of $m_\nu$ ($m^\nu_{12}=m^\nu_{13}$, $m^\nu_{22}=m^\nu_{33}$, $m^\nu_{11}+m^\nu_{12} = m^\nu_{22} +m^\nu_{23}$), see e.g. ref.~\cite{Altarelli:2005yx}, that can be accounted for by discrete symmetries (with a highly non-trivial construction needed to achieve a consistent and complete picture, including quarks and the charged fermion hierarchies). In the light of recent data, such models require sizeable corrections from the charged lepton sector~\cite{Frampton:2004ud,Romanino:2004ww,ref:ChLepCorr1,Hochmuth:2007wq,Goswami:2009yy,ref:ChLepCorr2,Marzocca:2011dh,Marzocca:2013cr,Girardi:2013sza,Petcov:2014laa}, making the TBM scheme as predictive as simple models without correlations (see however refs.~\cite{Marzocca:2011dh,Marzocca:2013cr,Girardi:2013sza,Petcov:2014laa} for a a possible prediction for the CP phase). 

In the following, we will concentrate in particular on the charged fermion (lepton) mass matrices, which are particularly suited for our approach due to the significant hierarchy among their eigenvalues.\footnote{With an abuse of mathematical terminology, we will use ``eigenvalues'' to mean ``singluar values".} This makes unlikely that the small eigenvalues arise as a consequence of accidental correlations among much larger quantities, an important element in our analysis, and is a sign of a non-anarchical origin of its matrix entries.
We will see that our approach allows to draw interesting conclusions on their contribution to lepton mixing. 

The precise formulation of our assumption will be given in Section~\ref{sec:stability}. Let us see here in a qualitative and intuitive way how assuming the absence of certain special correlations among matrix elements can translate into relevant information on the structure of the fermion mass matrices, using a simple and well known 2 family example: the charged lepton mass matrix $M_E$, restricted to the second and third families,
\begin{equation*}
\label{eq:example}
M_E = 
\begin{pmatrix}
M_{22} & M_{23} \\
M_{32} & M_{33}
\end{pmatrix}  
\end{equation*}
(throughout this paper we will use a ``RL'' convention for the charged fermion mass matrices). 
Suppose that $M_E = U^T_{e^c} M^\text{diag}_E U_e$, where $M^\text{diag}_E = \diag(m_\mu,m_\tau)$ and $U_e$, $U_{e^c}$ are rotations by angles $\theta$, $\theta_c$, respectively, that are both large, $\tan\theta \sim\tan\theta^c \sim 1$. As a consequence, all the four entries of $M_E$ are of the same order of magnitude as the tau mass $m_\tau$, and the observed relative smallness of $m_\mu$ is a consequence of a precise correlation among those four entries,
\begin{equation}
\label{eq:examplecorr}
M_{22} M_{33} - M_{23} M_{32} = 0,
\end{equation}
up to small corrections of relative order $\ord{m_\mu/m_\tau}$. Such a correlation can certainly occur, accidentally or as a consequence of a non-abelian symmetry.
But if we assume that it does not, this translates into constraints on the structure of the matrix $M_E$. Since $m_\tau \sim |M_{33}|$ (see Appendix~\ref{sec:ordering}) and $m_\tau m_\mu = |M_{33}M_{22}-M_{23}M_{32}|$, we  have in fact
\begin{equation}
\label{eq:mutau}
\frac{m_\mu}{m_\tau} = \frac{m_\mu m_\tau}{m_\tau^2} \sim \fracwithdelims{|}{|}{M_{33}M_{22} - M_{23}M_{32}}{M_{33}^2} .
\end{equation}
Requiring, according to our assumption, that the smallness of $m_\mu$ does not result from a fine-tuned cancellation among two correlated terms $M_{33}M_{22}$ and $M_{23}M_{32}$ (as in \eq{examplecorr}), we conclude that 
\begin{equation}
\label{eq:known}
\fracwithdelims{|}{|}{M_{22}}{M_{33}} \lesssim \frac{m_\mu}{m_\tau} \quad\text{and}\quad 
\fracwithdelims{|}{|}{M_{23}M_{32}}{M^2_{33}} \lesssim \frac{m_\mu}{m_\tau}, 
\end{equation}
which provides relevant information on the structure of the $m_E$ matrix. 

Interestingly, the above conditions can equivalently be obtained by requiring that the lightest eigenvalue $m_\mu$, or equivalently the product $m_\mu m_\tau = |\det M|$, is stable with respect to small variations of the matrix entries. The stability of an anomalously small quantity $X(a)$ with respect to a small variation $\Delta a\ll a$ of the variable $a$ is measured by the quantity 
\begin{equation}
\label{eq:D}
\Delta_a = \left| \frac{\Delta X}{\Delta a} \frac{a}{X} \right| \approx \fracwithdelims{|}{|}{\Delta \log X}{\Delta \log a} .
\end{equation}
In the $\Delta a \to 0$ limit, the definition above coincides with the ``fine-tuning'' or ``sensitivity'' parameter often used to measure the naturalness of the Higgs mass (for reasons that will become clear later, we prefer to keep a finite form here). The larger is $\Delta_a$, the more unstable is the smallness of $X$. When $a$ is assumed to be an independent fundamental parameter of the theory, it is desirable to have $\Delta_a \lesssim 1$, in such a way that the smallness of $X(a)$ can be considered ``natural'', i.e.\ not accidental. In the case of our $2\times 2$ mass matrix $M$, we can require that the small quantity $m_2 m_3 = |\det M|$, or $\det M$ itself, is stable with respect to variations of the matrix elements $M_{ij}$ and calculate the corresponding sensitivity parameters: 
\begin{equation}
\label{eq:D2x2}
\Delta_{M_{33}} = \Delta_{M_{22}} = \frac{|M_{22} M_{33}|}{m_2 m_3}, \quad
\Delta_{M_{23}} = \Delta_{M_{32}} = \frac{|M_{23} M_{32}|}{m_2 m_3}.
\end{equation}
Therefore, the assumption in \eq{known} is equivalent to imposing
\begin{equation}
\label{eq:equivalent}
\Delta_{M_{ij}} \lesssim 1
\end{equation}
for (one or) all the entries $M_{ij}$, $ij = 1,2$, and is therefore nothing but a stability assumption, at least if the parameters $M_{ij}$ can be considered independent. 

The arguments above on the structure of our toy $2\times 2$ lepton mass matrix are well known and underlie textures that have been widely considered in the literature. For example, textures with $M_{32}, M_{33} \sim m_\tau$, $M_{22}, M_{23} \sim m_\mu$ have been considered since ref.~\cite{Altarelli:1998ns} as possible explanations for the origin of the large atmospheric angle. The purpose of this paper is to  analyse in a rigorous and complete way the consequence on the structure of a full $3\times 3$ hierarchical mass matrix of the systematic application of the above ideas. 

In Sections~\ref{sec:stability} and~\ref{sec:structure} we precisely define the assumption we make, which generalises \eq{known}, and we study its connection with the absence of correlations in the full determinant and $2\times 2$ sub-determinants of $3\times 3$ fermion matrices. We also give different characterisations of stable mass matrices valid for any $n\times n$ matrix. This Section will make use of a number of useful results on mass matrices collected in the Appendices. In Section~\ref{sec:implications} we will consider examples of applications of our results to the charged lepton contributions to neutrino mixing. In particular, we will revisit the issue of whether the charged lepton contribution can account for all neutrino mixing and show that this is indeed possible without fine-tuning. We will also consider the case in which the charged lepton mass matrix combines with a maximal 12 rotation originating in the neutrino sector and we will see that this also leads to a plausible texture for the lepton mass matrix. In Section~\ref{sec:conclusions} we summarise our results. 

\section{The stability assumption}
\label{sec:stability}

In this Section we define the assumption we make in this paper, in the general case of a $n\times n$ matrix $M$, and we study its consequences, including an explicit equivalent  formulation in terms of constraints on products of matrix elements, which is the basis of the analysis carried out in the next Sections. The proofs of the statements in this Section are given in Appendix~\ref{sec:proofs}. 

Let $M$ be a generic complex $n\times n$ matrix with hierarchical eigenvalues
\begin{equation}
\label{eq:hierarchical}
0 < m_1 \ll \ldots \ll m_n ,
\end{equation}
representing for example a Dirac fermion mass matrix. Throughout this paper we will assume that its eigenvalues are stable in size with respect to small variations of the matrix elements $M_{ij}$. In order to give a precise definition of this assumption, it is useful to define the quantities
\begin{equation}
\label{eq:Sigma}
\Pi_p \equiv \Big(\sum_{k_1<\ldots <k_p} 
m^2_{k_1}\ldots m^2_{k_p} \Big)^{1/2}\approx m_n \ldots m_{n-p+1} ,
\end{equation}
where $p=1\ldots n$. For hierarchical eigenvalues, $\Pi_p$ is essentially the product of the $p$ largest eigenvalues, as shown in \eq{Sigma}. The quantities $\Pi_p$ are useful because, on the one hand, the requirement of the stability of the eigenvalues $m_1,\ldots, m_n$ can be equivalently formulated in terms of the stability of the products $m_n\ldots m_{n-p+1}\approx \Pi_p$;\footnote{Strictly speaking the two requirements are equivalent if $n$ is not too large, say $n \leq 3$. If $n\gg 1$, the stability of all $\Pi_p$ implies the stability of all $m_k$, but not viceversa. This can be seen by observing that $\Delta (\log \Pi_p)/\Delta (\log M_{ij}) \approx  \Delta (\log m_{n}) /\Delta (\log M_{ij}) + \ldots + \Delta (\log m_{n-p+1}) /\Delta (\log M_{ij})$. Therefore, even if the individual eigenvalues have sensitivities of order one, the sensitivity of $\Pi_p$ can be large, for large $p$ and $n$, because of the large number of $\ord{1}$ contributions. On the contrary, a small sensitivity for all $\Pi_p$ guarantees a small sensitivity for all the eigenvalues. Inverting the previous relations one finds in fact: $\Delta (\log m_{n}) /\Delta (\log M_{ij}) \approx \Delta (\log \Pi_1)/\Delta (\log M_{ij})$ and  $\Delta (\log m_{k}) /\Delta (\log M_{ij}) \approx \Delta (\log \Pi_{n-k})/\Delta (\log M_{ij}) - \Delta (\log \Pi_{n-k+1})/\Delta (\log M_{ij})$ for $k < n$.} on the other hand, the quantities $\Pi^2_p$ have a polynomial expression in terms of the matrix elements $M_{ij}$ and their conjugated, see \eq{identity}, which allows to translate the stability requirement into constraints on the matrix elements. 

\begin{definition}[stability assumption]
We say that the mass matrix $M$ is stable with respect to small variations of its matrix elements iff
\begin{equation}
\label{eq:assumption}
\left|
\frac{\Delta \Pi_p}{\Delta M_{ij}} 
\frac{M_{ij}}{\Pi_p} 
\right| \lesssim 1
\quad
\text{for } |\Delta M_{ij}| \ll |M_{ij}| 
\text{ and }i,j,p = 1\ldots n .
\end{equation}
\end{definition}
\noindent As explained, the definition above expresses the stability of the determination of the eigenvalues of $M$ (more precisely the products in \eq{Sigma}) with respect to small variation of any matrix entry. 

\begin{proposition}[relation with fine-tuning] 
\label{p1}
The stability assumption implies 
\begin{equation}
\label{eq:assumption2}
\left|
\frac{\partial \Pi_p}{\partial M_{ij}} 
\frac{M_{ij}}{\Pi_p} 
\right| \lesssim 1
\quad
\text{for } i,j,p = 1\ldots n ,
\end{equation}
but the viceversa is true only for $n=1,2$. 
\end{proposition}
\noindent An example of $3\times 3$ matrix $M$ that satisfies \eq{assumption2} but not \eq{assumption} is given in the Example~\ref{ex:1} in Appendix~\ref{sec:proofs}. 
The reason why \eq{assumption2} in that case misses the instability is that the latter does not show up when $|\Delta M_{ij}|$ is much smaller than the second eigenvalue (which is always the case in \eq{assumption2}, where the limit $\Delta M_{ij}\to 0$ is taken). This is the reason why we chose to use a definition of stability using finite differences. 

We now show that for $n\leq 3$ the stability assumption translates in practice into simple constraints on products of matrix entries, which correspond to the absence of cancellations in the expressions entering the determinants and sub-determinants of $M$. The constraints in \eqs{2x2} and~(\ref{eq:3x3}) are all we need for the analysis carried out in the next Sections. 
\begin{proposition}[characterization of stable matrices with $n\leq 3$]~
\label{p2}
\begin{enumerate}
\item
For $n=1$, $M$ is trivially always stable; 
\item
For $n=2$, $M$ is stable if and only if
\begin{equation}
\label{eq:2x2}
|M_{11} M_{22}| \lesssim m_1 m_2,
\qquad
|M_{12} M_{21}| \lesssim m_1 m_2 ;
\end{equation}
or equivalently if and only if $|M_{ij}M_{ji}| \lesssim m_i m_j$ for all $i,j=1,2$;
\item
For $n=3$, $M$ is stable if and only if
\begin{equation}
\label{eq:3x3}
\begin{gathered}
|M_{ih} M_{jk}| \lesssim m_2 m_3 \quad \text{for all $i\neq j$, $h\neq k$} \\
|M_{1i} M_{2j} M_{3k}| \lesssim m_1 m_2 m_3 \quad \text{for all $ijk$ permutations of $123$} .
\end{gathered}
\end{equation}
\end{enumerate}
\end{proposition}
\noindent The interpretation of the above characterisation is clear in the light of the results on mass matrices in Appendix~\ref{sec:results}. In particular, \eq{2x2} can be interpreted as the absence of cancellations in the RHS of $m_1 m_2 = |M_{11}M_{22} - M_{12} M_{21}|$, as discussed in the Introduction. As for the $n=3$ case, an analogous interpretation is possible in the light of the the fact that the absolute value of the determinant of any $p\times p$ submatrix of $M$ (in the case of \eq{3x3}, the $2\times 2$ submatrix made of the $M_{ih}, M_{jk}, M_{ik}, M_{jh}$ elements, with determinant $M_{ih} M_{jk} - M_{ik} M_{jh}$) is always smaller or equal to the product of the $p$ largest eigenvalues (in the case of \eq{3x3}, the product $m_2 m_3$). Moreover, $m_1 m_2 m_3 = |\sum_{ijk\text{ perm.\ of }123} M_{1i} M_{2j} M_{3k}|$, so that the last condition in \eq{3x3} can also be interpreted as the absence of cancellations in the previous expression for $m_1 m_2 m_3$. 

Note that the connection outlined above between the stability of $M$ and the absence of cancellations in the determinant and sub-determinants, although intuitive, is not trivial. For example, it does not hold for $n\geq 4$, as shown by the Example~\ref{ex:2} in Appendix~\ref{sec:proofs}. 

For completeness, we also give two additional characterisations of stable hierarchical matrices that emerge in the proof of the previous proposition. Let us first fix a matrix element $M_{ij}$ and define $\hat M_{(ij)}$ to be the matrix obtained from $M$ by setting to zero all the elements in the row $i$ and column $j$ except $M_{ij}$ and $\check M_{(ij)}$ the matrix with the element $ij$ set to zero, as in \eq{hat}. Let us also fix $1\leq p \leq n$ and denote by $\hat \Pi_{(ij)p}$ and $\check \Pi_{(ij)p}$ the quantities in \eq{Sigma} associated to $\hat M_{(ij)}$ and $\check M_{(ij)}$ respectively. 
\begin{proposition}[general characterisation of stable matrices] 
\label{p3}
The following three statements are equivalent:
\begin{enumerate}
\item \Eq{assumption} holds for given $p,i,j$;
\item $\hat \Pi_{(ij)p} \lesssim \Pi_p$;
\item $\check \Pi_{(ij)p} \lesssim \Pi_p$. 
\end{enumerate}
Therefore the stability of the mass matrix is equivalent to requiring 2.\ or 3.\ for all $i,j,p$.\end{proposition}
\noindent The intuitive meaning of the points 2.\ and 3.\ above has again to do with stability, as they state that setting to zero one of the matrix entries (or alternatively all the entries on the same row and column except that one) does not give rise to a drastic change of the structure of the eigenvalues. 

Appendices~\ref{sec:results} and~\ref{sec:proofs} contain a number of additional results, as well as the proofs of the statements in this Section. 

\section{General structure of stable charged fermion (lepton) mass matrices}
\label{sec:structure}

In this Section, we will  describe the general structure of a $3\times 3$ hierarchical fermion mass matrix satisfying the stability assumption, i.e.\ such that the hierarchy of its eigenvalues does not require accidental or dynamical correlations among its entries. 

Let us start with a remark on the ordering of rows and columns of $M$: it is always possible to order the rows and columns of $M$ in such a way that the structure of the matrix follows the hierarchy of the eigenvalues, i.e.\ in such a way that the third row and column are associated to the third and largest eigenvalue, and so on. More precisely, it is possible to order the rows and columns of $M$ in such a way that
\begin{equation}
\label{eq:ordering}
\begin{aligned}
|M_{33}| &= \ord{m_3}, \\[0.4mm]
|\det M_{[23]}| &= \ord{m_2 m_3}, \quad\text{and of course} \\
|\det M| &= m_1 m_2 m_3,
\end{aligned}
\end{equation}
where $M_{[23]}$ is the $2\times 2$ sub-matrix of $M$ corresponding to the second and third rows and columns (as in eqs.~(\ref{eq:minors}) and~(\ref{eq:square})). We will assume that this it the case in the following. 

En passant, one can wonder how far from $m_3$ and $m_2 m_3$ can $|M_{33}|$ and $|\det M_{[23]}|$ get in the equations above, or what exactly $\ord{m_3}$ and $\ord{m_2 m_3}$ mean. In Appendix~\ref{sec:ordering} we show that we can always make $|M_{33}| \gtrapprox m_3/\sqrt{3}\approx 0.6 \, m_3$ and $|\det M_{[23]}| \gtrapprox m_2 m_3/\sqrt{6} \approx 0.4\, m_2 m_3$. If $M$ did not satisfy the stability assumption (but is hierarchical), the bounds above would be qualitatively different, $|M_{33}| \gtrapprox m_3/3$ and $|\det M_{[23]}| \gtrapprox m_2 m_3/6$. 

Once the rows and columns of $M$ have been ordered as above, a stable $M$ is subject to the following constraints:
\begin{itemize}
\item
$|M_{3i}|,|M_{i3}|\leq m_3$, $i=1,2,3$;
\item
$|M_{2i}|, |M_{i2}| \lesssim m_2$, $i=1,2$;
\item
$|M_{11}| \lesssim m_1$;
\item
$|M_{ij}M_{ji}| \lesssim m_i m_j$ for all $i,j=1,2,3$ except $ij=13,31$;
\item
$|M_{13} M_{31}| \lesssim m_2 m_3$;
\item
$|M_{13} M_{32}|, |M_{23} M_{31}| \lesssim m_2 m_3$; 
\item
$|M_{13} M_{21} M_{32}|, |M_{31} M_{12} M_{23}|, |M_{13} M_{22} M_{31}| \lesssim m_1 m_2 m_3$.  
\end{itemize}
Viceversa, an hierarchical $M$ satisfying the constraints above (and having $m_1,m_2,m_3$ as eigenvalues) is automatically stable. 

While in the $2\times 2$ case $M$ satisfies the stability assumption iff $|M_{ij}M_{ji}| \lesssim m_i m_j$ for all $i,j=1,2$, in the $3\times 3$ case the corresponding constraint turns out to be true for all $i,j=1,2,3$ except for $ij=13,31$. We can then consider, in turn, two ranges for $|M_{13} M_{31}|$: $|M_{13} M_{31}| \lesssim m_1 m_3$ and the (somewhat less expected) $m_1 m_3 \ll |M_{13} M_{31}| \lesssim m_2 m_3$. In this second case, which we consider first, the structure of $M$ turns out to be particularly constrained. 

\subsection{$m_1 m_3 \ll |M_{13} M_{31}| \lesssim m_2 m_3$}

In this case, the constraints above force $|M_{22}| \ll m_2$, so that $|\det M_{[23]}| = \ord{m_2 m_3}$ must be accounted for by $|M_{23}M_{32}| \sim m_2 m_3$. The general structure of $M$ can then be described in terms of the size of the product $|M_{13}M_{31}|$,
\begin{equation}
\label{eq:pars}
k \equiv \frac{|M_{13}M_{31}|}{m_1 m_3}, 
\end{equation}
and in terms of the asymmetry, or degree of ``lopsidedness'', between $|M_{32}|$ and $|M_{23}|$ ($R_{23}$) and between $|M_{31}|$ and $|M_{13}|$ ($R_{13}$, or $R_{12} = R_{13}/R_{23}$),
\begin{equation}
\label{eq:say}
R_{23} \equiv \sqrt{\fracwithdelims{|}{|}{M_{32}}{M_{23}}}, \quad
R_{12} R_{23} \equiv \sqrt{\fracwithdelims{|}{|}{M_{31}}{M_{13}}} .
\end{equation}
The matrix $|M|$ of absolute values of the entries of $M$ has then the following structure
\begin{equation}
\label{eq:Mcase1}
|M| = \begin{pmatrix}
\lesssim m_1 &
\lesssim \sqrt{m_1 m_2}/(R_{12} \sqrt{k}) &
\sqrt{m_1 m_3 k} / (R_{12}  R_{23}) \\[1mm]
\lesssim \sqrt{m_1 m_2} R_{12}/\sqrt{k} &
\lesssim m_2/k &
\sqrt{m_2 m_3} / R_{23} \\[1mm]
\sqrt{m_1 m_3 k} \, (R_{12}  R_{23}) &
\sqrt{m_2 m_3} \, R_{23} &
\sim m_3
\end{pmatrix}, 
\end{equation}
where
\begin{equation}
\label{eq:ranges1}
1 \ll k \lesssim \frac{m_2}{m_1}, \quad 
\sqrt{\frac{m_2}{m_3}} \lesssim R_{23} \lesssim \sqrt{\frac{m_3}{m_2}} , \quad
\sqrt{\frac{m_1}{m_2}k} \lesssim R_{12}  \lesssim \sqrt{\frac{m_2}{m_1}\frac{1}{k}} .
\end{equation}
The largest stable values of $k$, $k\sim m_2/m_1$, require
\begin{equation}
\label{eq:Mcase1bis}
|M| = \begin{pmatrix}
\lesssim m_1 &
\lesssim m_1 &
\sim \sqrt{m_2 m_3} / R_{23} \\[1mm]
\lesssim m_1 &
\lesssim m_1 &
\sim \sqrt{m_2 m_3} / R_{23} \\[1mm]
\sim \sqrt{m_2 m_3} \, R_{23} &
\sim \sqrt{m_2 m_3} \, R_{23} &
\sim m_3
\end{pmatrix}, 
\end{equation}
where the lopsideness factor $R_{23}$ is bounded as in \eq{ranges1}. 

\subsection{$|M_{13} M_{31}| \lesssim m_1 m_3$}

In this case, $|M_{ij}M_{ji}| \lesssim m_i m_j$ holds for all $i,j=1,2,3$. A general parameterisation similar to equation \eq{Mcase1} is still possible, although it turns out to be more complicated. The lopsidedness parameters $R_{ij}$ can be defined only if the corresponding $|M_{ij}M_{ji}|$ is non-zero. If that is the case, we can define
\begin{equation}
\begin{gathered}
\label{eq:pars2}
k_{ij} \equiv \frac{|M_{ij}M_{ji}|}{m_i m_j}, \quad i\leq j, \quad \text{and}\quad
R_{ij} \equiv \sqrt{\fracwithdelims{|}{|}{M_{ji}}{M_{ij}}}, \quad i < j . 
\end{gathered}
\end{equation}
In terms of the above parameters we can then write 
\begin{equation}
\label{eq:Mcase2}
|M| = \begin{pmatrix}
\sqrt{k_{11}} \, m_1 &
\sqrt{m_1 m_2 k_{12}}/R_{12} &
\sqrt{m_1 m_3 k_{13}} / R_{13} \\[1mm]
\sqrt{m_1 m_2 k_{12}} \, R_{12} &
\sqrt{k_{22}} \, m_2 &
\sqrt{m_1 m_2 k_{23}} / R_{23} \\[1mm]
\sqrt{m_1 m_3 k_{13}} \, R_{13}  &
\sqrt{m_1 m_2 k_{23}} \, R_{23} &
\sqrt{k_{33}} \, m_3
\end{pmatrix}, 
\end{equation}
where 
\begin{equation}
\label{eq:conditions}
\begin{gathered}
k_{ij} \lesssim 1, \\
\sqrt{\frac{m_i}{m_j} k_{ij}} \lesssim R_{ij} \lesssim \sqrt{\frac{m_j}{m_i}\frac{1}{k_{ij}}}
 \\
\sqrt{\frac{m_1}{m_2} k_{23}k_{13}} \lesssim \frac{R_{13}}{R_{23}} \lesssim \sqrt{\frac{m_2}{m_1} \frac{1}{k_{23}k_{13}}} , \quad
\sqrt{k_{12}k_{23}k_{13}} \lesssim \frac{R_{23}R_{12}}{R_{13}} \lesssim \frac{1}{\sqrt{k_{12}k_{23}k_{13}}} . 
\end{gathered}
\end{equation}
The formulas above also apply to the previous case, and thus become general, provided that the constraint $k_{13} \lesssim 1$ is generalised to $k_{13} \lesssim m_2/m_1$ and provided that $k_{13} \sqrt{k_{22}}\lesssim 1$. 

\section{Examples}
\label{sec:implications}

\subsection{Can neutrino mixing arise from the charged lepton sector?}

As an example of applications of the above results, in this subsection we revisit the issue of whether the PMNS matrix can be dominated by the charged lepton contribution. The PMNS matrix $U$ is given by $U=U_e U^\dagger_\nu$, where $U_e$ and $U_\nu$ enter the diagonalisation of the charged lepton and neutrino mass matrices, $M_E = U^T_{e^c} M^\text{diag}_E U_e$, $M_\nu = U_\nu^T M^\text{diag}_\nu U_\nu$. Let us consider the possibility that $U_\nu$ is diagonal and all the mixing comes from the charged lepton sector, $U = U_e$ (up to phases that can be set to zero without loss of generality). 

We first observe that in such a case the last row of the charged lepton mass matrix $M_E$ is approximately determined by the PMNS matrix, as
\begin{equation}
\label{eq:line3a}
|M^E_{3i}| = |U_{3i}| m_\tau + \ord{m_\mu^2/m_\tau}~,
\end{equation}
where, experimentally, $|U_{3i}| = \ord{1}$.\footnote{In order to prove the previous equation, we first observe that $|U^e_{3i}| = |U_{3i}| = \ord{1}$ and $|U^{e^c}_{33}| = \ord{1}$ (because $|M^E_{33}| \sim m_\tau$) and therefore $|M^E_{3i}| = |U^{e^c}_{33} U^e_{3i}m_\tau| + \ord{m_\mu} \sim m_\tau$. The stability condition then implies $|M^E_{j3}| \lesssim m_\mu$ and $|U^{e^c}_{3j}| \lesssim m_\mu/m_\tau$, $j=1,2$ (since $|M^E_{j3}| = |U^{e^c}_{3j} U^e_{33}m_\tau| + \ord{m_\mu}$). Finally, unitarity implies $|U^{e^c}_{33}| = 1 - \ord{m_\mu/m_\tau}^2$ and $|U^{e^c}_{k3}| \lesssim m_\mu/m_\tau$, $k=1,2$. Therefore, $|M^E_{3i}| = |U^{e^c}_{33} U^e_{3i}m_\tau + U^{e^c}_{k3} U^e_{ki}m_k| = |U^{e^c}_{33} U^e_{3i}| m_\tau + \ord{m^2_\mu/m_\tau} = |U^e_{3i}| m_\tau + \ord{m^2_\mu/m_\tau}$. }

By using \eq{line3a} and the results for normal hierarchy from the global fit in ref.~\cite{Capozzi:2013csa} we then get the $1\sigma$ ranges 
\begin{equation}
\label{eq:line3b}
|M_E| \approx 
\begin{pmatrix}
\ldots & \ldots & \ldots \\[2mm]
\ldots & \ldots & \ldots \\[2mm]
 (\text{0.28--0.45}) \, m_\tau & (\text{0.50--0.62}) \, m_\tau & (\text{0.72-0.76}) \, m_\tau
\end{pmatrix},
\end{equation}
up to corrections suppressed by $(m_\mu/m_\tau)^2 \approx 0.003$. 

We now want to determine the constraints on the first and second lines that follow from the stability assumption. Using the characterisation of stable mass matrices in Section~\ref{sec:structure},
we find that we find that a lepton mass matrix $M_E$ in the form \eq{line3b} satisfies the stability assumption iff it is possible to find a $k$ such that 
\begin{equation}
\label{eq:MEfull}
|M_E| = 
\begin{pmatrix}
\lesssim m_e & \lesssim m_e & \lesssim k m_e \\[2mm]
\lesssim m_\mu/k & \lesssim m_\mu/k & \lesssim m_\mu \\[2mm]
\sim m_\tau & \sim m_\tau & \sim m_\tau
\end{pmatrix} , \quad \text{ with }
\quad
1\lesssim k \lesssim \frac{m_\mu}{m_e} .
\end{equation}

The above matrix can be diagonalised perturbatively with a series of $2\times 2$ unitary transformations, giving
\begin{equation}
\label{eq:UallE}
U = U_e = \Phi R_{12}(\theta'_{12},\phi') R_{23}(\theta^e_{23},\phi_3-\phi_2) R_{12}(\theta^e_{12},\phi_2-\phi_1),
\end{equation}
where $R_{ij}(\theta,\phi)$ denotes the $3\times 3$ unitary transformation consisting in the embedding of 
\begin{equation}
\label{eq:R12}
\begin{pmatrix}
\cos\theta & -\sin\theta e^{i\phi} \\
\sin\theta e^{-i\phi} & \cos\theta
\end{pmatrix}
\end{equation}
in the $ij$ block of the $3\times 3$ matrix; $R_{23}(\theta^e_{23},\phi_3-\phi_2)$ and $R_{12}(\theta^e_{12},\phi_2-\phi_1)$ are the rotations necessary to bring the third row of $M_E$ in diagonal form and are determined by that row, 
\begin{equation}
\label{eq:line3c}
M_E = 
\begin{pmatrix}
\ldots & \ldots & \ldots \\[2mm]
\ldots & \ldots & \ldots \\[2mm]
e^{i\phi_1} s^e_{12}s^e_{23} \, m_\tau   & e^{i\phi_2} c^e_{12}s^e_{23} \, m_\tau & e^{i\phi_3} c^e_{23} \, m_\tau
\end{pmatrix};
\end{equation}
$R_{12}(\theta'_{12},\phi')$ diagonalises the 12 block after the previous two rotations have been applied; $\Phi$ is a diagonal matrix of phases. The results above hold up to corrections of relative order $m^2_\mu/m^2_\tau$. \Eqs{line3b} and~(\ref{eq:MEfull}) give
\begin{equation}
\label{eq:estimate}
\tan\theta^e_{23} \sim \tan\theta^e_{12} \sim 1 \quad \text{and} \quad \tan\theta'_{12} \sim 1/k .
\end{equation}

\begin{figure}[ph!]
\begin{center}
%
\fbox{\footnotesize Normal Ordering} \\[0.2cm]
\vspace*{-0.2cm}
\hspace*{-0.65cm} 
\begin{minipage}{0.5\linewidth}
\begin{center}
	\hspace{0.75cm} \mbox{\footnotesize (a)} \\[0.5mm]
	\includegraphics[width=70mm]{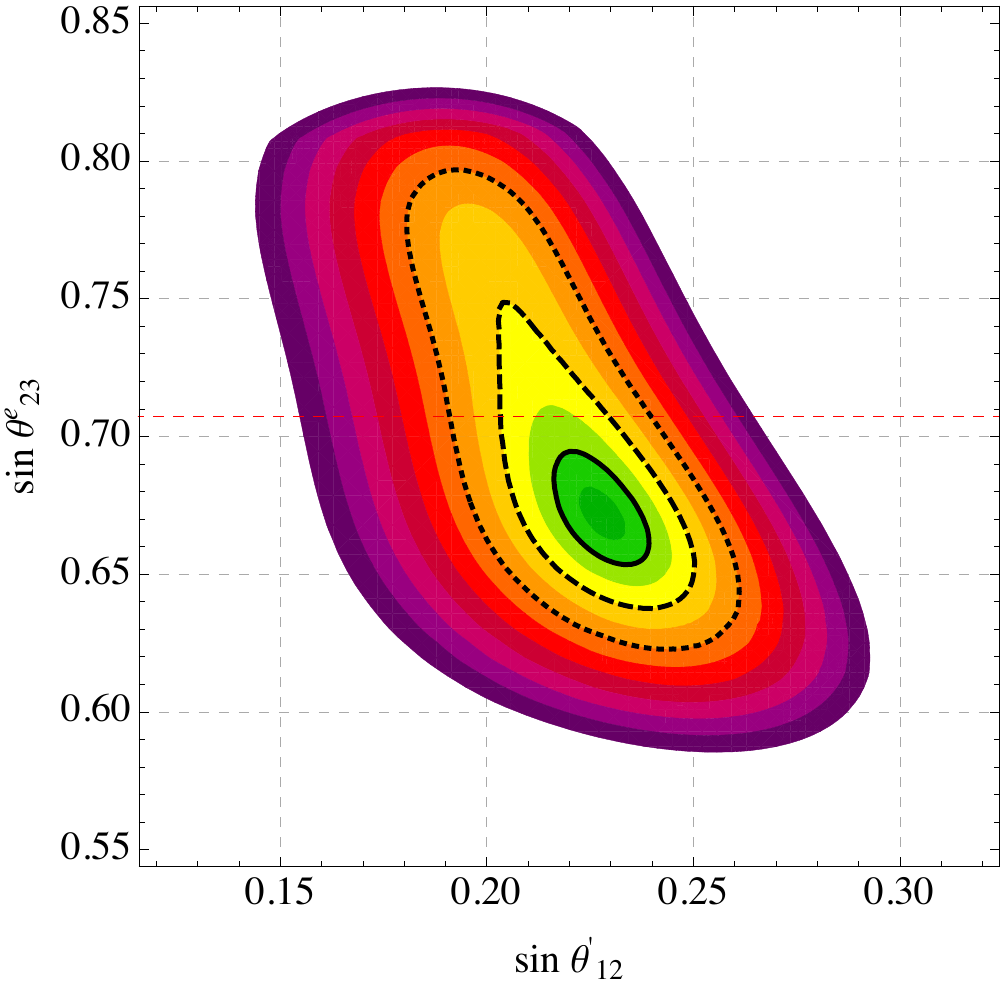}
\end{center}
\end{minipage}
%
%
\begin{minipage}{0.5\linewidth}
\begin{center}
	\hspace{0.5cm} \mbox{\footnotesize (b)} \\[0.5mm]
	\includegraphics[width=70mm]{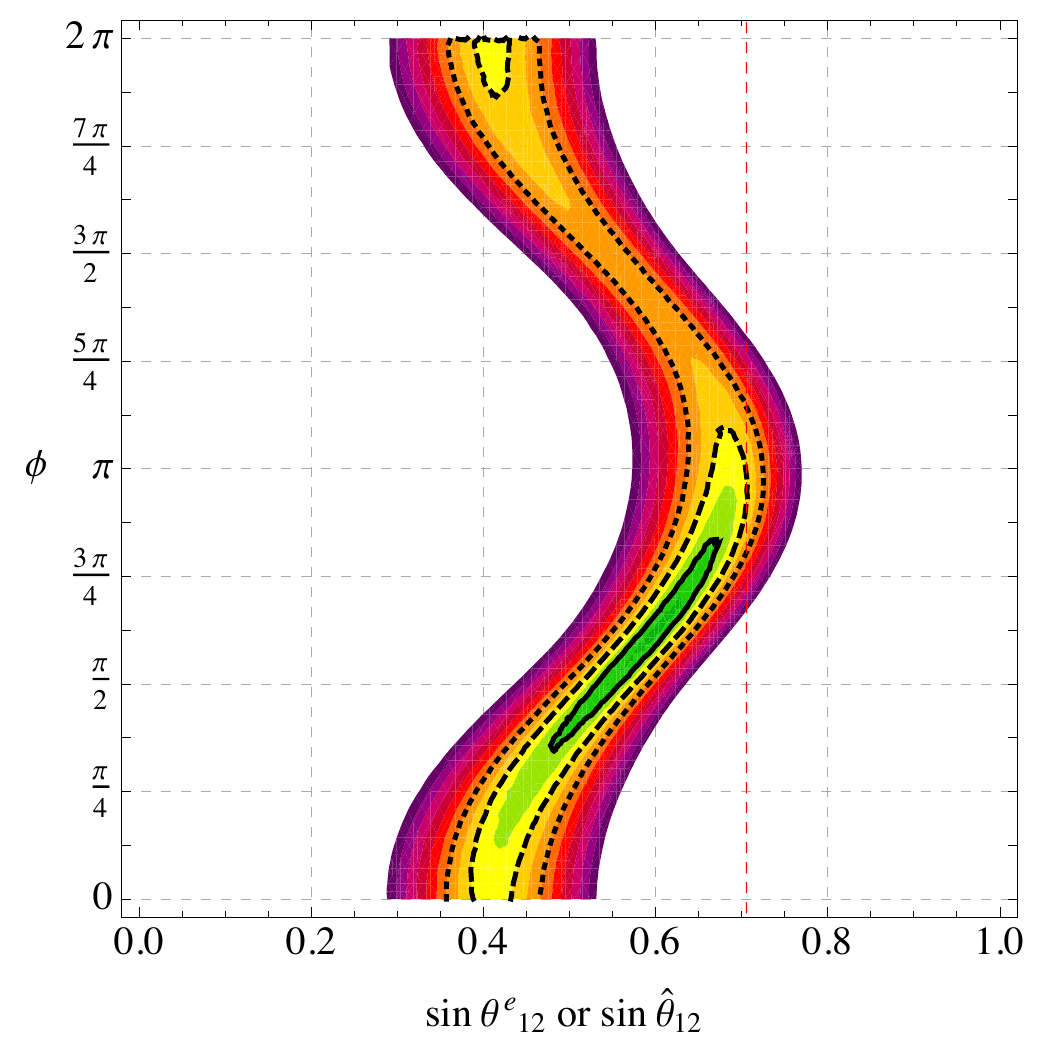}
	\end{center}
\end{minipage} 
\\
%
%
\fbox{\footnotesize Inverted Ordering} \\[0.2cm]
\vspace*{-0.2cm}
\hspace*{-0.65cm} 
\begin{minipage}{0.5\linewidth}
\begin{center}
	\hspace{0.75cm} \mbox{\footnotesize (c)} \\[0.5mm]
	\includegraphics[width=70mm]{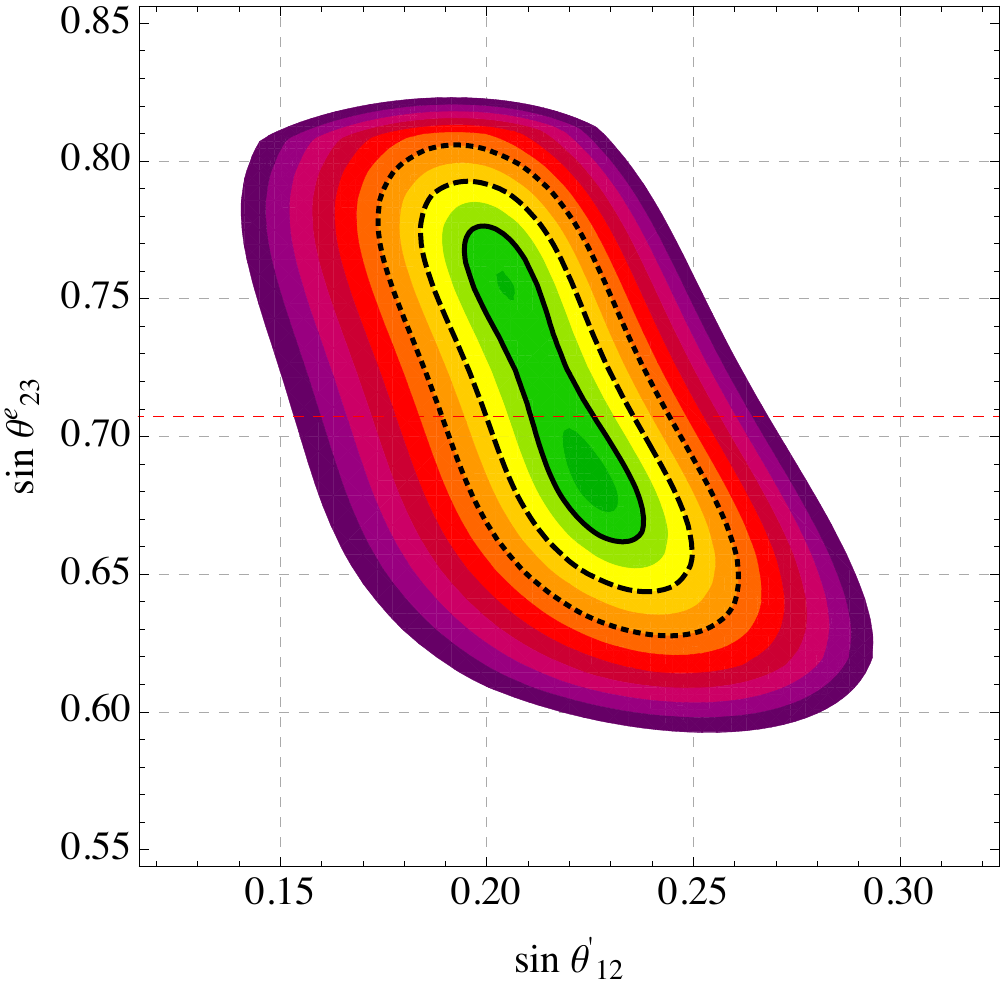}\
	\end{center}
\end{minipage}
%
%
\begin{minipage}{0.5\linewidth}
\begin{center}
	\hspace{0.5cm} \mbox{\footnotesize (d)} \\[0.5mm]
	\includegraphics[width=70mm]{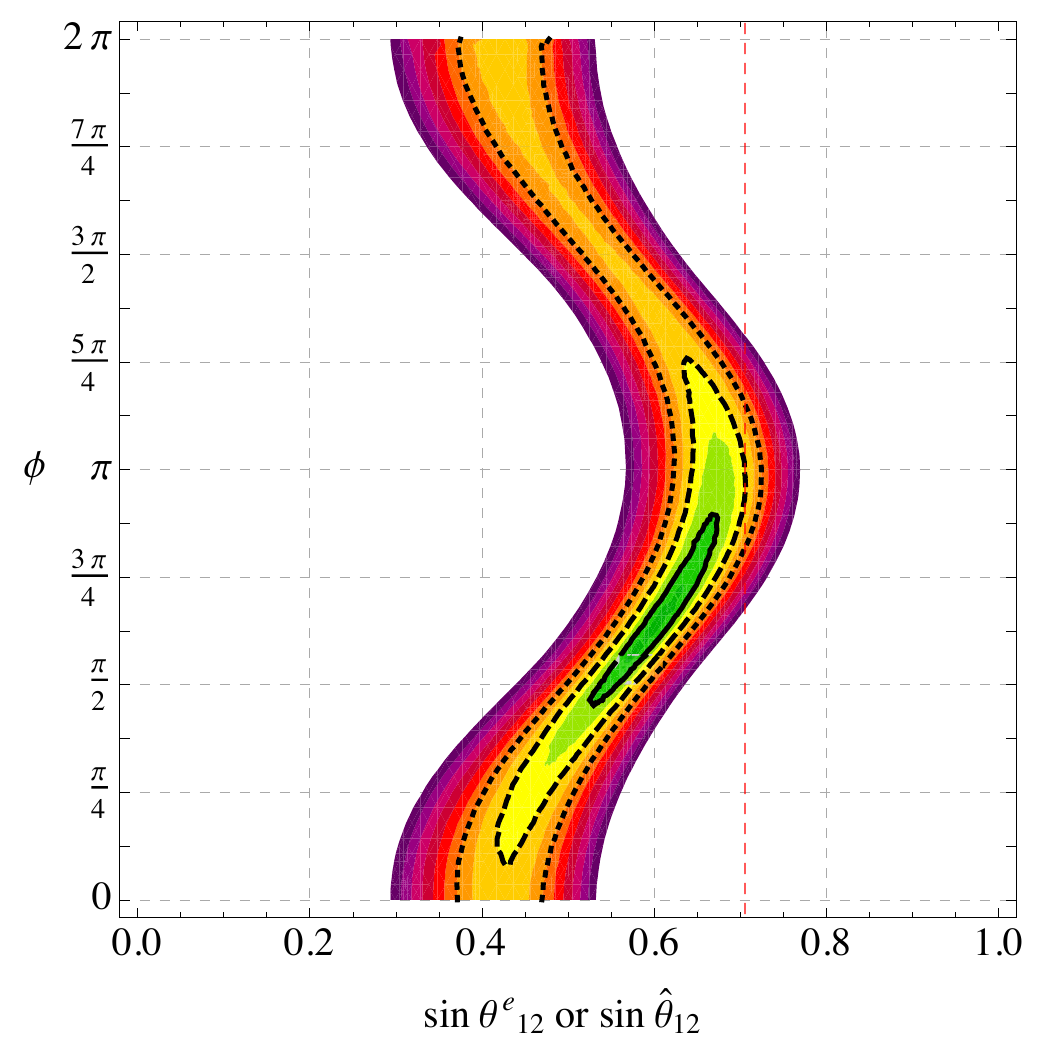}
	\end{center}
\end{minipage} 
\\
\vspace*{-0.2cm}
\begin{minipage}{0.5\linewidth}
\begin{center}
	\includegraphics[width=\textwidth]{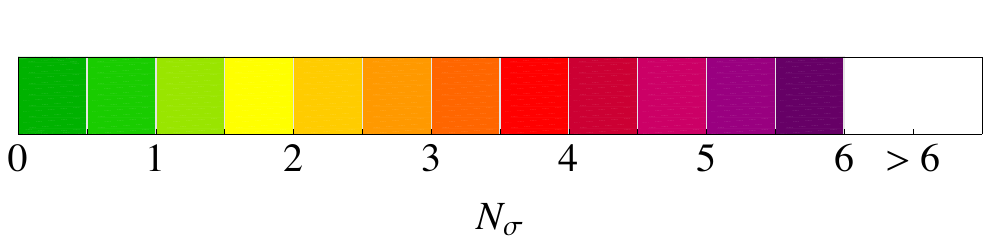}\\
\end{center}
\end{minipage}
\end{center}
\vspace*{-0.5cm}
\caption{\label{fig:fit_plot} 
\small Contours of $N_\sigma \equiv \sqrt{\Delta \chi^2}$ in the $(\sin \theta_{12}^\prime, \sin \theta^e_{23})$ (a,c) and $(\sin \theta_{12}^e, \phi)$ (b,d) planes. We construct the likelihood function using the results of the recent global fit of neutrino oscillation data from ref.~\cite{Capozzi:2013csa} for normal ordering (upper row) and inverted ordering (lower row) of neutrino masses. In plots (a,c) we use only the constraints on $\sin \theta_{13}$ and $\sin^2 \theta_{23}$ and the first two equations in eq.~\eqref{eq:parameterizations}. In plots (b,d) we include also the constraints on $\sin^2 \theta_{12}$ and $\delta$ and use the third line of eq.~\eqref{eq:parameterizations} as well as the relation between $\phi$ and $\delta$ obtained by comparing the expressions for $J_{CP}$ in the two parametrizations (see ref.~\cite{Marzocca:2013cr} for the details), and we marginalize over $\sin \theta_{12}^\prime$ and $\sin \theta^e_{23}$. The same analysis can be applied also to the case discussed in Section~\ref{sec:theta12pi4}, see eq.~\eqref{eq:parameterizationspi4}, by substituting $\theta_{12}^e$ with $\hat{\theta}_{12}$.
}
\end{figure}

The PMNS matrix in \eq{UallE} is in a form that has been already considered in the literature~\cite{Altarelli:2004jb,Romanino:2004ww,Petcov:2004rk,Marzocca:2011dh,Marzocca:2013cr}. The precise relation between the parameters in \eq{UallE} and the parameters of the standard parameterisation can be found in refs.~\cite{Marzocca:2011dh,Marzocca:2013cr}. In our notations, 
\begin{equation}
\label{eq:parameterizations}
 \begin{aligned}
	\sin \theta_{13} &= \sin \theta'_{12} \sin \theta^e_{23} = \ord{1} \sin\theta_{23}/k \\[2mm]
	\sin^2 \theta_{23} &= \sin^2 \theta^e_{23} \frac{\cos^2 \theta'_{12}}{1 - \sin^2 \theta'_{12} \sin^2 \theta^e_{23} } \\
	\sin^2 \theta_{12} &= \frac{\left| \sin \theta^e_{12} \cos \theta'_{12} + e^{i \phi} \cos \theta^e_{12} \cos \theta^e_{23} \sin \theta'_{12} \right|^2}{1 - \sin^2 \theta'_{12} \sin^2 \theta^e_{23} } ,
\end{aligned} 
\end{equation}
where $\phi = \phi' + \phi_1 - \phi_2$. A fit for the parameters $\theta^e_{23}$, $\theta^e_{12}$ and $\theta'_{12}$, $\phi$ is shown in \Fig{fit_plot}, using the results of the global fit of neutrino oscillation data from ref.~\cite{Capozzi:2013csa} both for normal and inverted ordering of neutrino masses. The $\ord{1}$ factor in the first equation is not expected to be small, unless a correlation among the entries of $M^E_{[23][12]}$ makes its determinant correspondingly small~\cite{Altarelli:2004jb}. If this is not the case, we estimate 
\begin{equation}
\label{eq:k}
1/k = \ord{1} \times 0.16.
\end{equation}
From Fig.~\ref{fig:fit_plot}(a,c) we also note also that, as a consequence of the first equation in~(\ref{eq:parameterizations}), the rotation angle $\theta'_{12}$ that diagonalises the 12 sector of $M_E$ has the same size, within errors, as the Cabibbo angle. Such a connection with the quark sector can be realised in the context of grand-unification \cite{Antusch:2011qg,Marzocca:2011dh,Antusch:2012fb,Antusch:2013kna}.

In the light of what above, we observe that:
\begin{itemize}
\item 
A small $\theta_{13}$ in the range
\begin{equation}
\label{eq:t13range}
0.03 \approx \frac{m_e}{m_\mu}\sin_{23} \lesssim \sin\theta_{13} \leq \sin_{23} \approx 0.7 ,
\end{equation}
including the measured range, can be obtained without the need of cancellations even if all neutrino mixing comes from the charged lepton sector.\footnote{In ref.~\cite{Altarelli:2004jb}, a small $\theta_{13}$ was associated to cancellations in the determinant of the $M^E_{[23][12]}$ submatrix, but it was also shown that the latter could be a natural consequence of a heavy vector-like lepton exchange dominance.}
\item
Independent of whether all neutrino mixing is accounted for by  the charged lepton contribution or not, the latter contribution is usually written as a product of two rotations in the ``standard order'' $U_e = R_{12} R_{23}$. We see that the ``inverted order'', $U_e = R_{23} R_{12}$, considered e.g.\ in refs.~\cite{Frampton:2004ud,Marzocca:2013cr}, can also be obtained (up to corrections of order $m_e/m_\mu$), without the need of correlations, when $1/k$ is at the lower end of its range, $1/k \sim m_e/m_\mu$. 
\item
The value of $k$ in \eq{k} is compatible with $k\sim \sqrt{m_\mu/m_e}$. Lepton mixing can therefore be accounted for in this set up by 
\begin{equation}
\label{eq:MEfull2}
|M_E| \sim 
\begin{pmatrix}
m_e &  m_e &  \sqrt{m_e m_\mu} \\[2mm]
\sqrt{m_e m_\mu} & \sqrt{m_e m_\mu} & m_\mu \\[2mm]
 m_\tau &  m_\tau &  m_\tau
\end{pmatrix} . 
\end{equation}
\end{itemize}

Finally, let us briefly discuss whether an abelian flavour model, for example, can account for the texture in \eq{MEfull2}. Often abelian models lead to textures in the form $M^E_{ij} \sim c_{ij} \lambda^c_i \lambda_j\,m_0$, with $0 < \lambda_i, \lambda^c_j < 1$ and $|c_{ij}| \sim 1$~\cite{Altarelli:2002hx,Altarelli:2004za}. Such textures can also be obtained in partial compositeness models (for a recent review see e.g. ref.~\cite{KerenZur:2012fr}). Clearly such textures can account for all the entries of the  above texture except for $M^E_{33}$, which parametrically would be expected to be $\ord{m_\tau \sqrt{m_\mu/m_e}}$ rather than $\ord{m_\tau}$, i.e.\ an order of magnitude larger. Still, a texture in the form $M^E_{ij} \sim c_{ij} \lambda^c_i \lambda_j\,m_0$ with $|M^E_{33}| = \ord{m_\tau \sqrt{m_\mu/m_e}}$ is not obviously ruled out. In fact, the parametric difference between the ratio $|M^E_{32}/M^E_{33}| \sim 0.07$ predicted by that texture and the ratio $|M^E_{32}/M^E_{33}| \sim 1$ in \eq{MEfull2} can be accounted by i) the fact that the precise observed value $|M^E_{32}/M^E_{33}| \approx 0.7$ is slightly smaller than 1, ii) the fact that in a two Higgs doublet model with large $\tan\beta$ the running of $|M^E_{32}/M^E_{33}|$ from a high scale to the electroweak scale can reduce its value by a factor 2 \cite{Antusch:2013jca}, and iii) a slightly stretched $\ord{1}$ factor. 

Another possibility is to consider an abelian flavour model with more than one flavon, which does not necessarily lead to a texture in the form $M^E_{ij} \sim c_{ij} \lambda^c_i \lambda_i\,m_0$. A complete example, also forcing the neutrino mass matrix to be diagonal, is provided in Appendix~\ref{sec:diagonalneutrinos}.

\subsection{Correction to $\theta_{12} = \pi/4$ from the charged lepton sector}
\label{sec:theta12pi4}

As a second example, let us consider the case in which the neutrino mass matrix contributes to lepton mixing with a maximal ``12'' rotation (up to phases),
\begin{equation}
\label{eq:nupi4}
U_\nu = \Phi_\nu R_{12}\fracwithdelims{(}{)}{\pi}{4} \Psi_\nu ,
\end{equation}
where $\Phi_\nu$ and $\Psi_\nu$ are diagonal matrices of phases. The charged lepton mass matrix must account in this case for the measured deviation of $\theta_{12}$ from $\pi/4$, besides for $\theta_{23}$ and $\theta_{13}$. 

As before we have $M_{3i} \approx m_\tau U^e_{3i}$, where now 
\begin{equation}
\label{eq:UUnu}
U_e = U U_\nu .
\end{equation}
We can still parameterize the last row of $M_E$ as in \eq{line3c}, with 
\begin{equation}
\label{eq:newte}
\begin{aligned}
s^e_{12} s^e_{23} \, e^{i\phi_1}  &= (\bar U_{31} \, e^{i\alpha_1} + \bar U_{32} \, e^{i\alpha_2})/\sqrt{2} \\
c^e_{12} s^e_{23} \, e^{i\phi_2} &= (-\bar U_{31} \, e^{i\alpha_1} + \bar U_{32} \, e^{i\alpha_2})e^{i\beta}/\sqrt{2} \\
c^e_{23} \, e^{i\phi_3}  &= \bar U_{33} \,e^{i\alpha_3} ,
\end{aligned}
\end{equation}
where we have denoted by $\bar U$ the PMNS matrix in the standard parameterization (the matrix 
$U$ in \eq{UUnu} is not necessarily in that parameterization). \Eqs{newte} show that the value of $\theta^e_{23}$ is still determined by the PMNS matrix to be in the $1\sigma$ range $0.72 < \cos\theta^e_{23} < 0.76$, while the value of $\theta^e_{12}$ also depends on the unknown phase $\alpha_1-\alpha_2$. A non zero value of $\theta^e_{12}$ is required in order to make $|U_{31}| \neq |U_{32}|$, as preferred by data at $2\sigma$ (see below). For the present central values of the PMNS parameters in ref.~\cite{Capozzi:2013csa} (normal hierarchy), one gets the lower bound $\tan\theta^e_{12} > 0.13$. While $\theta^e_{12}$ may be expected not to be far from this lower limit, large values are also allowed, provided that the relative phase $\alpha_1-\alpha_2$ in \eq{newte} is properly adjusted. 

In the light of what above, the texture for the third line of $M_E$ can be written as 
\begin{equation}
\label{eq:line3bpi4}
|M_E| \sim
\begin{pmatrix}
\ldots & \ldots & \ldots \\
\ldots & \ldots & \ldots \\
\epsilon\, m_\tau & m_\tau & m_\tau
\end{pmatrix},
\end{equation}
where $\epsilon = \tan\theta^e_{12}$ and indicatively we can consider the range $0.13 \lesssim \epsilon \lesssim 1$, with smaller values also allowed if PMNS parameters away from the best fit are considered (we will anyway assume that $\epsilon \gtrsim m_e/m_\mu \approx 0.005$, as indicated by present data). As the case $\epsilon = \ord{1}$ has been considered in the previous subsection, we are interested to the case in which $\epsilon$ is significantly smaller than one, but the discussion below holds in both cases. 

Let us now determine the constraints on the structure of the charged lepton mass matrix that follow from \eq{line3bpi4} and the stability assumption. Using the characterisation of stable mass matrices in Section~\ref{sec:structure}, we find that a lepton mass matrix $M_E$ in the form~(\ref{eq:line3bpi4}) satisfies the stability assumption iff it is possible to find a $k$ such that 
\begin{equation}
\label{eq:MEfullpi4}
|M_E| = 
\begin{pmatrix}
\lesssim m_e & \displaystyle \lesssim \frac{m_e}{\epsilon} \min(1,k\epsilon) & \lesssim m_e \, k \\[3mm]
\lesssim m_\mu/k & \displaystyle \lesssim\frac{m_\mu}{k\epsilon} \min(1,k\epsilon) & \sim m_\mu \\[4mm]
\sim\epsilon\, m_\tau & \sim m_\tau & \sim m_\tau
\end{pmatrix} , \quad \text{ with } \quad 1 \lesssim k \lesssim \frac{m_\mu}{m_e} . 
\end{equation}

We can now diagonalise the matrix in \eq{MEfullpi4} to obtain the charged lepton contribution to the PMNS matrix. A perturbative block by block diagonalisation gives as before 
\begin{equation}
\label{eq:UallEpi4}
U_e = \Phi R_{12}(\theta'_{12},\phi') R_{23}(\theta^e_{23},\phi_3-\phi_2) R_{12}(\theta^e_{12},\phi_2-\phi_1),
\end{equation}
where $\Phi$ is a diagonal matrix of phases, $R_{23}(\theta^e_{23},\phi_3-\phi_2)$ and $R_{12}(\theta^e_{12},\phi_2-\phi_1)$ are the rotations necessary to bring the third row of $M_E$ (parameterised as in \eq{line3c}) in diagonal form, $R_{12}(\theta'_{12},\phi')$ diagonalises the 12 block after the previous two rotations have been applied, and the result holds up to corrections of relative order $m^2_\mu/m^2_\tau$. \Eq{MEfullpi4} gives 
\begin{equation}
\label{eq:estimatepi4}
\tan\theta'_{12} \sim 1/k, \quad \tan\theta^e_{23} \sim 1, \quad \tan\theta^e_{12} =\epsilon .
\end{equation}

By combining $U_e$ in \eq{UallEpi4} with $U_\nu$ in \eq{nupi4} we find a PMNS matrix in the form
\begin{equation}
\label{eq:Upi4}
U = U_e U^\dagger_\nu = \Phi R_{12}(\theta'_{12},\phi') R_{23}(\theta^e_{23},\phi_3-\phi_2) R_{12}(\hat\theta_{12},\hat\phi_{12}) \Psi , 
\end{equation}
where $\Psi$ is a diagonal matrix of phases. The PMNS matrix is thus again in the form found in the previous subsection ($12\times 23\times 12$ rotations), but now the last 12 rotation $R_{12}(\theta^e_{12},\phi_2-\phi_1)$ is replaced by the combination of that rotation with the maximal 12 rotation provided by the neutrino sector
\begin{equation}
\label{eq:hatthetadef}
R_{12}(\hat\theta_{12},\hat\phi_{12}) = R_{12}(\theta^e_{12},\phi_2-\phi_1)R_{12}(\pi/4,\phi^\nu_{12})\times\text{phases},
\end{equation}
where $\phi^\nu_{12}$ is a combination of the phases in $\Phi_\nu$, $\Psi_\nu$. In the absence of phases, $\hat\theta_{12} = \pi/4\pm \theta^e_{12}$. In general, 
\begin{equation}
\label{eq:hattheta}
\frac{\pi}{4}-\theta^e_{12} \leq \hat\theta_{12} \leq \frac{\pi}{4} + \theta^e_{12} ,
\end{equation}
with $\hat\theta_{12}$, $\hat\phi_{12}$ given by 
\begin{equation}
\label{eq:hatthetaresult}
e^{i\hat\phi_{12}}\tan\hat\theta_{12} = e^{i\phi^\nu_{12}}
\frac{1+\tan\theta^e_{12}e^{i(\phi^e_{12}-\phi^\nu_{12})}}{1-\tan\theta^e_{12}e^{i(\phi^e_{12}-\phi^\nu_{12}})} .
\end{equation}

The PMNS matrix is again parameterised in the way considered e.g.\ in ref.~\cite{Marzocca:2013cr} in  terms of the angles $\theta'_{12}$, $\theta^e_{23}$ and $\hat\theta_{12}$ in \eq{Upi4} and of the phase $\phi = \phi' =\hat\phi_{12}$. The angles $\theta'_{12}$, $\theta^e_{23}$, $\hat\theta_{12}$ are related to the parameters of the charged lepton mass matrix in \eq{MEfullpi4} by 
\begin{equation}
\label{eq:estimatepi4bis}
\tan\theta' \sim 1/k, \quad \tan\theta^e_{23} \sim1, \quad
\pi/4-\epsilon \lesssim \tan\hat\theta_{12} \lesssim \pi/4 + \epsilon 
\end{equation}
and are related to the standard PMNS parameters by \eqs{parameterizations} with $\theta^e_{12} \to \hat\theta_{12}$ and $\phi = \phi'-\hat\phi_{12}$,
\begin{equation}
\label{eq:parameterizationspi4}
 \begin{aligned}
	\sin \theta_{13} &= \sin \theta'_{12} \sin \theta^e_{23} = \ord{1} \sin\theta_{23}/k \\[2mm]
	\sin^2 \theta_{23} &= \sin^2 \theta^e_{23} \frac{\cos^2 \theta'_{12}}{1 - \sin^2 \theta'_{12} \sin^2 \theta^e_{23} } \\
	\sin^2 \theta_{12} &= \frac{\left| \sin \hat\theta_{12} \cos \theta'_{12} + e^{i \phi} \cos \hat\theta_{12} \cos \theta^e_{23} \sin \theta'_{12} \right|^2}{1 - \sin^2 \theta'_{12} \sin^2 \theta^e_{23} } .
\end{aligned} 
\end{equation}
The determination of the PMNS parameters in \Figs{fit_plot} therefore still applies. In particular, the determination of $\theta^e_{23}$ and $\theta'_{12}$ is still given by \Fig{fit_plot}(a,c), while $\hat\theta_{12}$ and $\phi$ are determined by \Fig{fit_plot}(b,d). From \Fig{fit_plot}(b,d) we see that $\theta^e_{12} = 0$, corresponding to $\hat\theta_{12} = \pi/4$, is $2\sigma$ away from the best fit. Note also that the rotation $\theta'_{12}$ in the 12 sector of $M_E$ has again the same size as the Cabibbo angle. 

Note that two factors, both associated to the charged lepton sector, contribute to make $\theta_{12}$ different from the maximal value provided by the neutrino sector. One is the $\theta^e_{12}$ rotation induced by $M^E_{31}$, which makes $\hat\theta_{12} \neq \pi/4$, and the other is the $\theta'_{12}$ rotation used to diagonalise the 12 block of $M^E_{12}$ after the other two blocks have been diagonalised. It has been observed~\cite{Marzocca:2013cr} that in the absence of the $\theta^e_{12}$ contribution, i.e.\ when $\hat\theta_{12} = \pi/4$, the $\theta'_{12}$ rotation alone can account for the deviation of $\theta_{12}$ from $\pi/4$ only at the price of a $2\sigma$ tension (as $\theta'_{12}$ is constrained by $\theta_{13}$, see \eqs{parameterizationspi4}). Here we see that this tension disappears if the independent contribution $\theta^e_{12}$, induced by $M^E_{31}$, is taken into account. In such a scheme, $\theta'_{12}$ determines $\theta_{13}$ and $\theta^e_{12}$ further contributes to the deviation of $\theta_{12}$ from the neutrino contribution. Summarizing:
\begin{itemize}
\item 
A small $\theta_{13}$ in the range 
\begin{equation*}
0.03 \approx \frac{m_e}{m_\mu}\sin_{23} \lesssim \sin\theta_{13} \leq \sin_{23} \approx 0.7 ,
\end{equation*}
can be again induced without fine-tuning by the rotation $\theta'_{12}$, whose natural size is set by $1/k$. The experimental value of $\sin\theta_{13} = \ord{1}\sin\theta_{23}/k$ gives $1/k = \ord{1}\times 0.16$. 
\item
The previous rotation alone can account for the deviation of $\theta_{12}$ from $\pi/4$ only at the price of a $2\sigma$ tension, with present data. On the other hand, this tension disappears if the independent contribution to $\theta_{12}$ induced by a non-zero ratio $\epsilon = |M^E_{31}/M^E_{32}|$  is taken into account. Therefore, a plausible and stable texture for the charged lepton mass matrix can account at the same time for the atmospheric mixing angle, the $\theta_{13}$ angle, and the deviation of the $\theta_{12}$ angle from $\pi/4$. 
\end{itemize}

Finally, we comment on the possible origin of the texture in \eq{MEfullpi4}. We observe that the latter is compatible with a form $M^E_{ij} \sim c_{ij} \lambda^c_i \lambda_j\, m_0$, with $0 < \lambda_i, \lambda^c_j< 1$ and $|c_{ij}| \sim 1$, provided that $\epsilon \lesssim 1/k \sim 0.16$. Together with the experimental $2\sigma$ bound $\epsilon \gtrsim 0.13$, this implies $\epsilon \sim 1/k$. The structure $M^E_{ij} \sim c_{ij} \lambda^c_i \lambda_j\, m_0$ and the constraint $\det M_E = m_e m_\mu m_\tau$ then allow to rewrite \eq{MEfullpi4} as 
\begin{equation}
\label{eq:MEfullpi42}
|M_E| \sim 
\begin{pmatrix}
m_e & \displaystyle  \frac{m_e}{\epsilon} & \displaystyle \frac{m_e}{\epsilon} \\[3mm]
\epsilon\, m_\mu &  m_\mu &  m_\mu \\[4mm]
\epsilon\, m_\tau &  m_\tau &  m_\tau
\end{pmatrix} , \quad \epsilon \sim \text{0.13--0.16} . 
\end{equation}
The previous texture is indeed in the form $M^E_{ij} \sim  c_{ij} \lambda^c_i \lambda_i \, m_0$, with $(\lambda_1,\lambda_2,\lambda_3) \propto (\epsilon, 1, 1)$ and $(\lambda^c_1,\lambda^c_2,\lambda^c_3) \propto (m_e/\epsilon, m_\mu, m_\tau)$. It can also be written in the form $M^E_{ij} \sim c_{ij} \epsilon^{q^c_i+q_j}m_0$, with appropriate choice of $\epsilon$ and of the charges $q_i$, $q^c_i$. Explicit and complete flavour models will be considered elsewhere. 

\section{Summary}
\label{sec:conclusions}

We have studied general properties and specific examples of hierarchical fermion mass matrices satisfying a ``stability'' assumption. The latter amounts to assuming the stability of the smaller eigenvalues with respect to small perturbations of the matrix entries. Such an assumption is equivalent to the absence of certain precise correlations, be them accidental or forced by a dynamical/symmetry principle, among the matrix entries and is therefore also motivated by the fact that no evidence of special correlations has so far emerged from data. 

We have found a simple and general characterisation of a stable $3\times 3$ mass matrix $M$ with eigenvalues $m_i$, $i=1,2,3$, in terms of products of matrix entries that proves useful for practical applications, 
\begin{equation*}
\begin{gathered}
|M_{ih} M_{jk}| \lesssim m_2 m_3 \quad \text{for all $i\neq j$, $j\neq k$} \\
|M_{1i} M_{2j} M_{3k}| \lesssim m_1 m_2 m_3 \quad \text{for all $ijk$ permutations of $123$} .
\end{gathered}
\end{equation*}
A number of exact relations involving the minors of $M$ obtained in the appendices show that the latter corresponds to the absence of cancellations in the expressions entering the determinants and sub-determinants of $M$. 

As an example of application of the general results, we have revisited the issue of the the charged lepton contribution to neutrino mixing and determined the structure of the charged lepton mass matrix under two assumption for the neutrino contribution: i) no contribution at all (all mixing from the charged lepton sector) and ii) it only provides a maximal $\theta_{12}$ angle. 

In the first case, we have seen that lepton mixing can indeed all come from the charged lepton sector and that this does not need to fine-tune the value of $\theta_{13}$, as long as $\theta_{13} \gtrsim m_e/m_\mu \sin\theta_{23}  \approx 0.03$, as it turned out to be. We have also translated the present determination of the standard PMNS parameters into a determination of alternative, equivalent parameters, directly related to the charged lepton matrix entries. The latter determination also allows to determine with good accuracy the whole third row of the charged lepton mass matrix. We have also briefly discussed the possible origin of the textures we have considered. 

In the case in which the neutrino sector only provides a maximal 12 rotation, we have shown that present data provide a $2\sigma$ evidence for a non-vanishing $31$ entry of the charged lepton mass matrix. The PMNS matrix turns out in fact to be given by a product of 12 and 23 rotations, $U = 12_1 \times 23 \times 12_2 \times 12_{\pi/4}$, where the neutrino sector only provides for the last one. Both the first and the second 12 rotations contribute to shift $\theta_{12}$ from $\pi/4$. The first one is the rotation used to diagonalise the 12 block of $M^E$ after the other two blocks have been diagonalised and is directly related to $\theta_{13}$. The second one is induced by a non zero value of $M^E_{31}/M^E_{32}$. Sometimes only the first one is considered, with the second set to zero. In such a case, a $2\sigma$ tension arises between the value of the 12 rotation needed to account for $\theta_{13}$ and the value needed to account for the deviation from $\theta_{12} = \pi/4$ (also due to the constraints on the phase $\delta$). On the other hand, the tension disappears if the second 12 rotation is taken into account. In such a case, the first 12 rotation determines $\theta_{13}$ and the independent second rotation further contributes to the deviation of $\theta_{12}$ from $\pi/4$. This way, a plausible texture for the charged lepton mass matrix can account at the same time for the atmospheric mixing angle, the $\theta_{13}$ angle, and the deviation of the $\theta_{12}$ angle from $\pi/2$. 

In both cases, the left-handed rotation that diagonalises the 12 sector of $M_E$ has the same size, within errors, as the Cabibbo angle, which may be considered as a hint in support of grand-unification.

Finally, independent of whether all neutrino mixing is accounted for by  the charged lepton contribution or not, we have shown that the so-called ``inverted order'' of the 12 and 23 rotations in the charged lepton sector, $U_e = R_{23} R_{12}$ can also be obtained without fine-tuning (up to corrections of order $m_e/m_\mu$).

\section*{Acknowledgments}
We thank Ferruccio Feruglio, Michele Frigerio and Serguey Petcov for useful discussions. The work of A.R.\ was supported by the ERC Advanced Grant no. 267985 ``DaMESyFla'' and by the European Union FP7 ITN ``Invisibles'' (Marie Curie Actions, PITN- GA-201-289442). 

\appendix

\section{Useful results}
\label{sec:results}

In this Appendix, we collect some results that have been used in the main text and will be used in Appendix~\ref{sec:proofs}. 

Let us first define some notations. Below, $M$ will denote a $n\times n$ generic complex matrix, possibly representing a fermion mass matrix. The matrix $M$ can be diagonalized by using two independent unitary matrices, 
\begin{equation}
\label{eq:diagonalization}
M = V^T M_D U, \quad U,V \in U(n), \quad M_D = \diag(m_1,\ldots,m_n),
\end{equation}
where $m_1,\ldots,m_n \geq 0$ are uniquely defined  singular values of $M$ (referred in the text as eigenvalues), ordered according to their sizes, $m_1 \leq \ldots \leq m_n$. We denote by $M_{[i_1\ldots i_p][j_a \ldots j_q]}$ the $p\times q$ sub-matrix made of the elements in the rows $i_1\ldots i_p$ and columns $j_1\ldots j_p$ of $M$,
\begin{equation}
\label{eq:minors}
\left( M_{[i_1\ldots i_p][j_1 \ldots j_q]} \right)_{ab} \equiv M_{i_aj_b} 
\end{equation}
($p,q = 1\ldots n$, $a=1\ldots p$, $b = 1\ldots q$). If the rows and columns coincide, we also use the notation 
\begin{equation}
\label{eq:square}
M_{[i_1\ldots i_p]} \equiv M_{[i_1\ldots i_p][i_1\ldots i_p]}. 
\end{equation}

A first useful result is the fact that the determinant of any squared $p\times p$ submatrix of $M$ is bound by the $p$ largest singular values of $M$,
\begin{equation}
\label{eq:supdet}
\left|\det M_{[i_1\ldots i_p][j_1 \ldots j_p]}\right| \leq m_n \ldots m_{n-p+1} .
\end{equation}
In the case $p=n$, the inequality~\eq{supdet} becomes of course an equality. For $p=1$, \eq{supdet} shows that all matrix elements are bound by the largest eigenvalue, $|M_{ij}| \leq m_n$. These inequalities are complementary to the ones in \eq{identity} below. 

\Eq{supdet} follows from a known result of linear algebra stating that the singular values $\hat m_1 \leq \ldots \leq \hat m_p$ of the $p\times p$ submatrix $M_{[i_1\ldots i_p][j_1 \ldots j_p]}$ are bound by the $p$ largest singular values of $M$, $\hat m_i \leq m_{n-p+i}$, $i=1\ldots p$, see e.g. ref.~\cite{R.A.Horn_and_C.R.Johnson:1991}.%
\footnote{It can also be obtained as follows. If two out of $i_1\ldots i_p$ are equal, \eq{supdet} is trivially verified. If $i_1\ldots i_p$ are all different, $\left|\det M_{[i_1\ldots i_p][j_1 \ldots j_p]}\right|^2 \leq \sum_{k_1<\ldots <k_p} \left|\det M_{[i_1\ldots i_p][k_1 \ldots k_p]}\right|^2 = \det [(M^\dagger M)_{[i_1\ldots i_p]}] 
= \det (U^\dagger M^2_D U)_{[i_1\ldots i_p]}  
= \sum_{k_1 < \ldots < k_p} |\det U_{[k_1\ldots k_p][i_1\ldots i_p]}|^2m^2_{k_1}\ldots m^2_{k_p} \leq m^2_{n} \ldots m^2_{n-p+1} \sum_{k_1 < \ldots < k_p} |\det U_{[k_1\ldots k_p][i_1\ldots i_p]}|^2 
= m^2_{n} \ldots m^2_{n-p+1} |\det (U^\dagger U)_{[i_1\ldots i_p]}|^2
= m^2_{n} \ldots m^2_{n-p+1}$. }
 
A related but independent result allows to obtain combinations of $p$ singular values through the determinant of $p\times p$ submatrices:
\begin{equation}
\label{eq:identity}
\Pi^2_p = \sum_{i_1 < \ldots < i_p} m^2_{i_1} \ldots m^2_{i_p} = \sum_{\substack{ h_1 < \ldots < h_p \\  k_1 < \ldots < k_p }}
\left|
\det M_{[h_1\ldots h_p][k_1\ldots k_p]}
\right|^2.
\end{equation}
The relation above generalizes the $p=1$ result $\sum_{i=1}^n m^2_i = \sum_{i,j=1}^n|M_{ij}|^2$ obtained in ref.~\cite{Frigerio:2002rd}. For $p=n$ it reduces to $m^2_1\ldots m^2_n = |\det M|^2$. The general case follows from equating the coefficients of $\lambda^{n-p}$ in the secular equation $\det(\lambda \mathbf{1} - M^\dagger M) = \prod_{i=1}^n (\lambda - m^2_i)$. The result is particularly useful in the case of hierarchical singular values $m^2_1 \ll \ldots \ll m^2_n$, in which case $\sum_{i_1 < \ldots < i_p} m^2_{i_1} \ldots m^2_{i_p} \approx m^2_n\ldots m^2_{n-p+1}$ and  \eq{identity} becomes an expression for the product of the $p$ largest squared singular values of $M$.

\section{Proofs of the results in Section~\ref{sec:stability}}
\label{sec:proofs}

We now prove the results stated in Section~\ref{sec:stability}, starting from Proposition~\ref{p3}, whose discussion is preparatory to the proof of the other two. In the following, and in the main text, $x\lesssim y$ ($x\gtrsim y$) indicates that $x < y$ ($x > y$) or $x$ is of the same order of $y$, i.e.\ they differ by a factor of order one. Therefore, $x\lesssim y$  ($x\gtrsim y$) is  equivalent to the negation of $x \gg y$ ($x\ll y$). Moreover, $a \lessapprox b$ ($a \gtrapprox b$) indicates that $a < b +\epsilon$ ($a \gtrapprox b - \epsilon$), with $0<\epsilon\ll |b|$. 

\subsection{Proof of Proposition~\ref{p3}}

For convenience, we remind that the proposition states that the following three statements are equivalent:
\begin{enumerate}
\item \Eq{assumption} holds for given $p,i,j$;
\item $\hat \Pi_{(ij)p} \lesssim \Pi_p$;
\item $\check \Pi_{(ij)p} \lesssim \Pi_p$. 
\end{enumerate}
We also remind that the quantities $\hat \Pi_{(ij)p}$ and $\check \Pi_{(ij)p}$ are associated to the mass matrices
\begin{gather}
\label{eq:hat}
\hat M_{(ij)} = 
\begin{pmatrix}
M_{11} & \cdots & 0 & \cdots & M_{1n} \\
\cdots & \cdots & 0 & \cdots & \cdots \\
M_{i-1,1} & \cdots & 0 & \cdots & M_{i-1,n} \\
0 & 0 & M_{ij} & 0 & 0 \\
M_{i+1,1} & \cdots & 0 & \cdots & M_{i+1,n} \\
\cdots & \cdots & 0 & \cdots & \cdots \\
M_{n1} & \cdots & 0 & \cdots & M_{nn}
\end{pmatrix}, \;
\check M_{(ij)} =
\begin{pmatrix}
M_{11} & \cdots & M_{1j} & \cdots & M_{1n} \\
\cdots & \cdots & \cdots & \cdots & \cdots \\
M_{i1} & \cdots & 0 & \cdots & M_{in} \\
\cdots & \cdots & \cdots & \cdots & \cdots \\
M_{n1} & \cdots & M_{nj} & \cdots & M_{nn}
\end{pmatrix}. 
\end{gather}
As the quantities $\Pi_p$ can be profitably calculated in terms of the determinant of sub-matrices (\eq{identity}), let us first determine the relation among the sub-determinants of $M$, $\hat M_{(ij)}$, $\check M_{(ij)}$. The relation depends on whether the sub-matrix includes the row $i$ and the column $j$. Accordingly, we have (for convenience, we fix $i,j$ and drop the suffix $(ij)$ in $\hat M$, $\check M$, $\hat\Pi_p$, $\check\Pi_p$) 
\begin{equation}
\label{eq:subdet}
\begin{aligned}
\det \hat M_{[ii_1\ldots i_{p-1}][jj_1\ldots j_{p-1}]} &= M_{ij} \,\det M_{[i_1\ldots i_{p-1}][j_1\ldots j_{p-1}]} \\
\det \hat M_{[ii_1\ldots i_{p-1}][j_1\ldots j_{p}]} &= 0 \\
\det\hat M_{[i_1\ldots i_{p}][jj_1\ldots j_{p-1}]} &= 0 \\
\det \hat M_{[i_1\ldots i_{p}][j_1\ldots j_{p}]} &= \det M_{[i_1\ldots i_{p}][j_1\ldots j_{p}]} \\[1mm]
\det \check M_{[ii_1\ldots i_{p-1}][jj_1\ldots j_{p-1}]} &= \det  M_{[ii_1\ldots i_{p-1}][jj_1\ldots j_{p-1}]} - M_{ij} \,\det M_{[i_1\ldots i_{p-1}][j_1\ldots j_{p-1}]} \\
\det \check M_{[ii_1\ldots i_{p-1}][j_1\ldots j_{p}]} &= \det M_{[ii_1\ldots i_{p-1}][j_1\ldots j_{p}]}  \\
\det \check M_{[i_1\ldots i_{p}][jj_1\ldots j_{p-1}]} &= \det M_{[i_1\ldots i_{p}][jj_1\ldots j_{p-1}]}  \\
\det \check M_{[i_1\ldots i_{p}][j_1\ldots j_{p}]} &= \det M_{[i_1\ldots i_{p}][j_1\ldots j_{p}]} .
\end{aligned}
\end{equation}
In the above equations, all $i_1\ldots i_p$ are different from $i$ and all $j_1\ldots j_p$ different from $j$. 

Let us begin proving that $2 \Rightarrow 1$. Using \eq{identity} one finds 
\begin{equation}
\label{eq:basic}
|M_{ij}| \frac{\Delta \Pi^2_p}{|\Delta M_{ij}|} = \sum_{\alpha\in I^p_{ij}}
\left(
e^{i\theta} v^*_\alpha w_\alpha + e^{-i\theta} v_\alpha w^*_\alpha + \fracwithdelims{|}{|}{\Delta M_{ij}}{M_{ij}} |w_\alpha|^2 
\right),
\end{equation}
where
\begin{equation}
\begin{gathered}
\label{eq:proof2}
e^{i\theta} = \frac{\Delta M_{ij}/M_{ij}}{|\Delta M_{ij}/M_{ij}|}, \quad 
I^{p}_{ij} = \left\{(i_1\ldots i_{p-1},j_1\ldots j_{p-1}) : 
\begin{matrix}
1 \leq i_1 <\ldots < i_{p-1} \leq n, \text{ all } \neq i  \\ 1 \leq j_1 <\ldots < j_{p-1} \leq n,  \text{ all } \neq j 
\end{matrix}
\right\} , \\[1mm]
v_{(i_1\ldots i_{p-1},j_1\ldots j_{p-1})} = \det M_{[i,i_1\ldots i_{p-1}][j,j_1\ldots j_{p-1}]} , \\
w_{(i_1\ldots i_{p-1},j_1\ldots j_{p-1})} = M_{ij} \,\det M_{[i_1\ldots i_{p-1}][j_1\ldots j_{p-1}]} = \det \hat M_{[i,i_1\ldots i_{p-1}][j,j_1\ldots j_{p-1}]} .
\end{gathered}
\end{equation}
For $p=1$, \eq{basic} should be interpreted as 
\begin{equation}
\label{eq:basic1}
|M_{ij}| \frac{\Delta \Pi^2_1}{|\Delta M_{ij}|} = 2 \cos\theta |M_{ij}|^2 + |\Delta M_{ij} M_{ij}|.
\end{equation}
Now, 
\begin{align}
\label{eq:proof3}
\sum_{\alpha\in I^{p}_{ij}} |v_\alpha|^2 & \leq \sum_{\substack{ i_1 < \ldots < i_p \\ j_1 < \ldots < j_p }}
|\det M_{[i_1\ldots i_p][j_1\ldots j_p]}|^2 = \Pi^2_p 
\quad \text{and} \quad \\
\sum_{\alpha\in I^{p}_{ij}} |w_\alpha|^2 & \leq \sum_{\substack{ i_1 < \ldots < i_p \\ j_1 < \ldots < j_p }}
|\det \hat M_{[i_1\ldots i_p][j_1\ldots j_p]}|^2 = \hat \Pi^2_p \lesssim \Pi^2_p ,
\end{align}
the last approximate inequality being the hypothesis. Because of the Cauchy-Schwartz inequality, we also have $|\sum_\alpha v^*_\alpha w_\alpha| 
\lesssim \Pi^2_p$. All in all, we have proven point 1, as 
\begin{equation}
\label{eq:cvd3}
\left|
\frac{\Delta \Pi_p}{\Delta M_{ij}} 
\frac{M_{ij}}{\Pi_p} 
\right| \approx 
\frac{1}{2} \left |\frac{\Delta \Pi^2_p}{\Delta M_{ij}} 
\frac{M_{ij}}{\Pi^2_p} \right| \leq
\frac{|\sum_{\alpha} v^*_\alpha w_\alpha |}{\Pi^2_p} + \frac{1}{2} \fracwithdelims{|}{|}{\Delta M_{ij}}{M_{ij}} \frac{\sum_{\alpha} |w_\alpha|^2}{\Pi^2_p} \lesssim 1
\end{equation}
for all $\Delta M_{ij}$ with $|\Delta M_{ij}| \ll |M_{ij}|$ (note that in the first step in \eq{cvd3} we have neglected a term of the same order of the sub-leading second term in the RHS). 

Let us now prove, by contradiction, that $1\Rightarrow 3$. Suppose that $\check\Pi_p \lesssim \Pi_p$ was not verified. Then we would have $\check\Pi_p = R \,\Pi_p$, with $R \gg 1$. The large size of $\check \Pi_p$ would then imply a large size of $\sum_{\alpha} |w_\alpha|^2$, as 
\begin{align}
\label{eq:size}
&\sum_{\alpha\in I^p_{ij}} |w_\alpha|^2 = 
\sum_{\substack{ i_1 < \ldots < i_{p-1}, \neq i \\ j_1 < \ldots < j_{p-1}, \neq j }}
|M_{ij} \det M_{[i_1\ldots i_{p-1}][j_i\ldots j_{p-1}]}|^2 
\geq \\
&\hfill \Big|
\sum_{\substack{i_1 < \ldots < i_{p-1}, \neq i \\ j_1 < \ldots < j_{p-1}, \neq j}}
\left(
|\det M_{[ii_1\ldots i_{p-1}][jj_i\ldots j_{p-1}]}|^2 - |\det \check M_{[ii_1\ldots i_{p-1}][jj_i\ldots j_{p-1}]}|^2
\right) \Big| 
= 
|\check\Pi^2_p -  \Pi^2_p| \approx R^2 \Pi^2_p  , \notag
\end{align}
where we have used \eqs{subdet}. 
Consider now a variation of $M_{ij}$ by 
\begin{equation}
\label{eq:variation}
\Delta M_{ij} = \frac{k}{R}M_{ij} e^{i\phi} ,
\end{equation}
where $k$ is a positive number of order one and $\phi$ is a phase chosen in such a way that $2 \re[e^{i\theta} v^*_\alpha w_\alpha] = 0$ in \eq{basic}. Then $|\Delta M_{ij}| \ll |M_{ij}|$, but \eq{basic} gives
\begin{equation}
\label{eq:basic2}
\left| \frac{M_{ij}}{\Pi^2_p} \frac{\Delta \Pi^2_p}{\Delta M_{ij}} \right| = 
\frac{k}{R}
\frac{\sum_{\alpha}
|w_\alpha|^2}{\Pi^2_p} \gtrsim k R \gg 1,
\end{equation}
which would contradict the assumption. 

Let us finally prove that $3\Rightarrow 2$. This can be done by observing that $\check \Pi_p \lesssim \Pi_p$ implies 
\begin{multline}
\label{eq:32}
\hat \Pi^2_p = 
\sum_{\substack{ i_1 < \ldots < i_{p}, \neq i \\ j_1 < \ldots < j_{p}, \neq j }}
|\det M_{[i_1\ldots i_{p}][j_i\ldots j_{p}]}|^2 + \\
\sum_{\substack{ i_1 < \ldots < i_{p-1}, \neq i \\ j_1 < \ldots < j_{p-1}, \neq j }}
|\det M_{[ii_1\ldots i_{p-1}][jj_i\ldots j_{p-1}]} - \det \check M_{[ii_1\ldots i_{p-1}][jj_i\ldots j_{p-1}]}|^2 \leq \\
\sum_{\substack{ i_1 < \ldots < i_{p}, \neq i \\ j_1 < \ldots < j_{p}, \neq j }}
|\det M_{[i_1\ldots i_{p}][j_i\ldots j_{p}]}|^2 + 
\sum_{\substack{ i_1 < \ldots < i_{p-1}, \neq i \\ j_1 < \ldots < j_{p-1}, \neq j }}
|\det M_{[ii_1\ldots i_{p-1}][jj_i\ldots j_{p-1}]}|^2 + \\
\sum_{\substack{ i_1 < \ldots < i_{p-1}, \neq i \\ j_1 < \ldots < j_{p-1}, \neq j }}
|\det \check M_{[ii_1\ldots i_{p-1}][jj_i\ldots j_{p-1}]}|^2 \leq
\Pi^2_p + \check \Pi^2_p \sim \Pi^2_p ,
\end{multline}
where we have used \eqs{subdet} to obtain the first equality. This proves point 2. 

\subsection{Proof of Proposition~\ref{p2}}

For convenience, we remind that this proposition characterises as follows the stability of matrices $M$ with dimension $n\leq 3$: 
\begin{enumerate}
\item
For $n=1$, $M$ is always stable; 
\item
For $n=2$, $M$ is stable if and only if
\begin{equation}
\label{eq:2x2A}
|M_{11} M_{22}| \lesssim m_1 m_2,
\qquad
|M_{12} M_{21}| \lesssim m_1 m_2 ;
\end{equation}
\item
For $n=3$, $M$ is stable if and only if
\globallabel{eq:3x3A}
\begin{gather}
|M_{ih} M_{jk}| \lesssim m_2 m_3 \quad \text{for all $i\neq j$, $j\neq k$} \mytag \\
|M_{1i} M_{2j} M_{3k}| \lesssim m_1 m_2 m_3 \quad \text{for all $ijk$ permutations of $123$} .\mytag 
\end{gather}
\end{enumerate}

Let us start observing that for $p=1$ (any $n$) \eq{basic1} gives
\begin{equation}
\label{eq:p1}
\left| \frac{M_{ij}}{\Pi_1} \frac{\Delta \Pi_1}{\Delta M_{ij}} \right| \approx \cos\theta \frac{|M_{ij}|^2}{\sum_{hk}|M_{hk}|^2} \leq 1. 
\end{equation}
This proves in particular that $M$ is always stable for $n=1$. 

Given what above, for $n=2$ we just need to consider the case $p=2$. In general, for $p=n$, \eq{basic} gives 
\begin{equation}
\label{eq:pn}
\left| \frac{M_{ij}}{\Pi_n} \frac{\Delta \Pi_n}{\Delta M_{ij}} \right| \approx
\left| \re\left[
e^{i\theta} \frac{M_{ij}\,\text{cof}\, M_{ij}}{\det M} 
\right] +
\bigg|
\frac{M_{ij}\,\text{cof}\, M_{ij}}{\det M} 
\right|^2 \fracwithdelims{|}{|}{\Delta M_{ij}}{M_{ij}} \bigg|.
\end{equation}
Therefore, $\text{LHS} \lesssim 1$ in the previous equation for all $\Delta M_{ij}$ (i.e.\ for all $\theta$)  if and only if 
\begin{equation}
\label{eq:pn2}
|M_{ij} \,\text{cof}\, M_{ij}| \lesssim m_1 \ldots m_n
\qquad \text{(or $M_{ij}M^{-1}_{ji} \lesssim 1$)}. 
\end{equation}
In the $p=n=2$ case the above relations coincide with the ones in \eq{2x2A}, which proves the case $n=2$. 

The proof of the $n=3$ case is more involved. First of all, let us show that the stability or $\Pi_2$ with respect to variation of any matrix element is equivalent to $|M_{ih} M_{jk}| \lesssim m_2 m_3$ for all $i\neq j$, $h\neq k$. It is easy to show that the stability of $\Pi_2$ implies the latter relations: proposition~\ref{p3} states that the stability of $\Pi_2$ implies $\hat \Pi_{(ij)2} \lesssim \Pi_2$; then $|M_{ih} M_{jk}| = |\det\hat M^{(ij)}_{[ij][hk]}| \leq \hat \Pi_{(ij)2} \lesssim \Pi_2\simeq m_2 m_3$. Viceversa, if $|M_{ih} M_{jk}| \lesssim m_2 m_3$ for all $i\neq j$, $h\neq k$, we have, using \eq{basic} as before,
\begin{multline}
\label{eq:n2}
\left| \frac{M_{ih}}{\Pi_2} \frac{\Delta \Pi_2}{\Delta M_{ih}} \right| \lesssim 
\left| \frac{\sum_{hk} (\det M_{[ij][hk]})^* M_{ih} M_{jk} }{m^2_2m^2_3} \right| \\ \lesssim 
\frac{\sum_{hk} |\det M_{[ij][hk]}| }{m_2m_3} 
\leq 
2\frac{(\sum_{hk} |\det M_{[ij][hk]}|^2)^{1/2}}{m_2m_3} \leq 2\frac{\Pi_2}{m_2m_3} \simeq 2, 
\end{multline}
which proves that $\Pi_2$ is stable.\footnote{In \eq{n2} we have used $\sum_{i=1}^k |x_i| \leq \sqrt{k} (\sum_{i=1}^k |x_i|^2 )^{1/2}$ to the sum of the 4 terms in $\sum_{hk} |\det M_{[ij][hk]}|$.} 

In order to complete the proof of the $n=3$ case, we now show that the stability of $\Pi_3$ is equivalent to $|M_{1i} M_{2j} M_{3k}| \lesssim m_1 m_2 m_3$ for  all $ijk$ permutations of $123$.  First, using again \eq{basic}, we find that the stability of $\Pi_3$ is equivalent to 
\begin{equation}
\label{eq:cof}
|M_{ij}\,\text{cof}\, M_{ij}| \lesssim m_1m_2m_3  \quad \text{for all $ij$} .
\end{equation}
We then have to show that \eq{cof} is equivalent to eq.~(\ref{eq:3x3A}b). It is easy to show that eq.~(\ref{eq:3x3A}b) implies \eq{cof}. In order to show that \eq{cof} implies eq.~(\ref{eq:3x3A}b), let us first observe that there must exist at least one $2\times 2$ sub-matrix $M_{[ij][hk]}$ with determinant $|\det M_{[ij][hk]}| = \ord{m_2m_3}$. Otherwise, if $|\det M_{[ij][hk]}| \ll m_2 m_3$ for all sub-matrices, we would also have $\Pi^2_2 = \sum_{i<j,h<k} |\det M_{[ij][hk]}|^2 \ll m^2_2 m^2_3$.\footnote{More precisely, we can show that there is at least one sub-matrix $M_{[ij][hk]}$ such that $\det M_{[ij][hk]}\geq m_2 m_3/2$. In order to show it, we anticipate that there can be at most 4 sub-determinants giving a sizeable contribution to $\Pi_2$ (see Appendix~\ref{sec:ordering}). Then for at least one of the 4 sizeable sub-determinants we must have $|M_{[ij][hk]}|^2 \geq \Pi^2_2/4 \approx m^2_2 m^2_3 /4$, i.e. $|\det M_{[ij][hk]}|\geq m_2 m_3/2$.} Without loss of generality, we can assume such sub-matrix to be $M_{[23]}$. Then \eq{cof} for $ij=11$ forces $|M_{11}| \lesssim m_1$. Since we also have $|M_{22}M_{33}| \lesssim m_2 m_3$, we conclude that 
\begin{equation}
\label{eq:123}
|M_{11}M_{22}M_{33}| \lesssim m_1 m_2 m_3.
\end{equation}
We have therefore proven one of the relations in eq.~(\ref{eq:3x3A}b). All the other ones follow because of the constraints \eq{cof}. For example, using \eq{cof} for $ij=33$, $|M_{11}M_{22}-M_{12}M_{21}| |M_{33}| \lesssim m_1 m_2 m_3$, 
we obtain
\begin{equation}
\label{eq:123bis}
|M_{12}M_{21}M_{33}| \lesssim m_1 m_2 m_3.
\end{equation}
Using \eq{cof} for $ij=12$, we obtain $|M_{12}M_{23}M_{31}| \lesssim m_1 m_2 m_3$. And so on and so forth (the 9 constraints in \eq{cof} are enough to constrain all the 6 products in eq.~(\ref{eq:3x3A}b). This completes the proof of the $n=3$ case and thus of Proposition~\ref{p2}. 

\subsection{Proposition~\ref{p2} cannot be extended to $n=4$}

As mentioned in the text, the characterisation in Proposition~\ref{p2} cannot be extended to the case $n=4$. For example, not all $n=4$ hierarchical matrices satisfying the stability assumption satisfy $|M_{1i}M_{2j}M_{3k}M_{4l}| \lesssim m_1m_2m_3m_4$ for all $ijkl$ permutations of $1234$. This is the case for example of the matrix in \eq{example2}. 
\begin{example} Consider the matrix
\label{ex:2}
\begin{equation}
\label{eq:example2}
M = \begin{pmatrix}
0 & \epsilon' & \epsilon' & \epsilon' \\
\epsilon' & 0 & 0 & 1\\
\epsilon' & \epsilon' & 0 & 1 \\
0 & \epsilon & \epsilon & 1
\end{pmatrix} ,
\end{equation}
where $\epsilon' \ll \epsilon \ll 1$. The singular values are approximately given by $\epsilon'(\epsilon'/\epsilon)$, $\epsilon'/2$, $\sqrt{4/3}\, \epsilon$, $\sqrt{3}$. The matrix satisfies the stability assumption. However, $M_{13}M_{24}M_{31}M_{42} \approx (\epsilon/\epsilon')\, m_1 m_2 m_3 m_4 \gg m_1 m_2 m_3 m_4$.
\end{example}

\subsection{Proof of Proposition~\ref{p1}}

Before illustrating the proof, let us define more precisely the quantity on the LHS of \eq{assumption2}. If the limit of the LHS of \eq{assumption} for $\Delta M_{ij} \to 0$ existed, we would simply have  
\begin{equation}
\label{eq:limit}
\left|\frac{\partial \Pi_p}{\partial M_{ij}} 
\frac{M_{ij}}{\Pi_p} 
\right|  = \lim_{\Delta M_{ij}\to 0} 
\left| \frac{\Delta \Pi_p}{\Delta M_{ij}} 
\frac{M_{ij}}{\Pi_p} 
\right| .
\end{equation}
On the other hand, the quantities $\Pi_p$ are not holomorphic functions of the variable $M_{ij}$, and the limit depends on the direction along which $\Delta M_{ij} \to 0$. In such a case, we replace the RHS of \eq{limit} by the  maximum value taken by the limit when $\Delta M_{ij}$ approaches 0 from different directions in the complex plane ($\Delta M_{ij} = \alpha z$, $z\in\mathbb{C}$, $|z| = 1$, $\alpha \in\mathbb{R}$, $\alpha\to 0$). Since $\Pi^2_p$ can be considered as an holomorphic function of $M_{ij}$ and $M^*_{ij}$ (through \eq{identity}), we have 
\begin{equation}
\label{eq:precisedef}
\max_{z} \lim_{\Delta M_{ij} = \alpha z \to 0} 
\left| \frac{\Delta \Pi_p}{\Delta M_{ij}} 
\frac{M_{ij}}{\Pi_p} 
\right| = 
\left|\frac{\partial \Pi_p}{\partial M_{ij}} 
\frac{M_{ij}}{\Pi_p} 
\right| +
\left|\frac{\partial \Pi_p}{\partial M^*_{ij}} 
\frac{M^*_{ij}}{\Pi_p} 
\right| .
\end{equation}
In short, we define the LHS in \eq{assumption2} as the quantity in \eq{precisedef}. 

Let us now prove Proposition~\ref{p1}, which therefore states that the stability assumption implies 
\begin{equation}
\label{eq:assumption2precise}
\left|\frac{\partial \Pi_p}{\partial M_{ij}} 
\frac{M_{ij}}{\Pi_p} 
\right| +
\left|\frac{\partial \Pi_p}{\partial M^*_{ij}} 
\frac{M^*_{ij}}{\Pi_p} 
\right| \lesssim 1
\quad
\text{for } i,j,p = 1\ldots n ,
\end{equation}
but the viceversa is true only for $n=1,2$. 

The fact that \eq{assumption} implies \eq{assumption2precise} simply follows from \eq{precisedef}. We then need to prove that the viceversa is true for $n=1,2$, but not for $n\geq 3$. 

For $p = 1$ (any $n$), both \eq{assumption} and \eq{assumption2precise} are always verified. This proves the viceversa for $n=1$ and $n=2$, $p=1$. For $p = n = 2$, it is easy to see (for example from \eq{basic}, maximising with respect to $\theta$) that 
\begin{equation}
\label{eq:pippa}
\left|\frac{\partial \Pi_2}{\partial M_{ij}} 
\frac{M_{ij}}{\Pi_2} 
\right| +
\left|\frac{\partial \Pi_2}{\partial M^*_{ij}} 
\frac{M^*_{ij}}{\Pi_2} 
\right|  =  \fracwithdelims{|}{|}{M_{ij}M_{hk}}{m_1 m_2},
\end{equation}
where $M_{hk}$ is the matrix element opposite to $M_{ij}$ in $M$. Therefore \eq{assumption2precise} implies \eq{2x2A}, which implies that $M$ is stable. This proves the viceversa for $n=2$. 

Finally, we need to prove that the viceversa is not true for $n=3$. This is illustrated by the following Example. 
\begin{example} Consider the matrix
\label{ex:1}
\begin{equation}
\label{eq:example1}
M = \begin{pmatrix}
\epsilon' & 1 & 1 \\
0 & \epsilon & 0 \\
0 & 1 & 1
\end{pmatrix} ,
\end{equation}
where $\epsilon' \ll \epsilon \ll 1$. The singular values are approximately given by $\epsilon'/\sqrt{2}$, $\epsilon/\sqrt{2}$, $2$. Using for example the general relation 
\begin{equation}
\label{eq:pappa}
\left|\frac{\partial \Pi_p}{\partial M_{ij}} 
\frac{M_{ij}}{\Pi_p} 
\right| +
\left|\frac{\partial \Pi_p}{\partial M^*_{ij}} 
\frac{M^*_{ij}}{\Pi_p} 
\right|  =  \frac{|\sum_{\alpha\in I^p_{ij}} v^*_\alpha w_\alpha|}{m_1 m_2},
\end{equation}
\end{example}
\noindent one can see that \eq{assumption2precise} is verified. On the other hand, $M$ does not satisfy the stability assumption because $M_{12} M_{33} \gg m_2 m_3$, which contradicts eq.~(\ref{eq:3x3A}a).

\section{Ordering rows and columns}
\label{sec:ordering}

In this Appendix we discuss the results on the ordering of rows and columns of a $3\times 3$ hierarchical mass matrix $M$ mentioned in Section~\ref{sec:structure}. 

Let us first consider a hierarchical matrix $M$ that does not necessarily satisfy the stability assumption. The following lemma proves useful to discuss this case. 
\begin{lemma*}[ordering for unitary matrices] 
\label{lemma1}
Given a $3\times 3$ unitary matrix $U$, it is possible to permute its columns (rows) in such a way that 
\begin{equation}
\label{eq:orderingU}
|U_{33}| \geq  \frac{1}{\sqrt{3}}, \qquad
|\det U_{[23]}| \geq \frac{1}{\sqrt{6}} .
\end{equation}
Moreover, it is not possible to set more stringent general bounds: for any $\epsilon > 0$ there exists a unitary matrix $U$ for which it is not possible to find an ordering such that $|U_{33}| \geq  1/\sqrt{3} +\epsilon$ and $|\det U_{[23]}| \geq 1/\sqrt{6}+\epsilon$. 
\end{lemma*}

\begin{proof} To prove the first bound in \eq{orderingU} it suffices to observe that $|U_{31}|^2+|U_{32}|^2+|U_{33}|^2 = 1$, so that $\max_{i}|U_{3i}|^2 \geq 1/3$. We can then permute the columns of $U$ in such a way that $|U_{33}| = \max_{i} |U_{3i}| \geq 1/\sqrt{3}$. Consider now an ordering in which $|U_{33}| \geq  1/\sqrt{3}$. As $|\det U_{[23][13]}|^2 + |\det U_{[23][23]}|^2 = |U_{23}|^2 + |U_{33}|^2 \geq 1/3$, the larger determinant will not be smaller than 1/6. We can then order the first two columns in such a way that $|\det U_{[23]}| = \max_{i=1,2} |\det U_{[23][i3]}| \geq 1/\sqrt{6}$. 

To prove that the bounds cannot be made more stringent, it suffices to consider the matrix
\begin{equation}
\label{eq:exampleU}
U = 
\begin{pmatrix}
\displaystyle
-\frac{1}{\sqrt{6}} - \frac{\epsilon}{\sqrt{2}} &
\displaystyle
-\frac{1}{\sqrt{6}} - \frac{\epsilon}{\sqrt{2}} &
\displaystyle
\sqrt{\frac{2}{3}} - \sqrt{2}\epsilon' \\
\displaystyle
\frac{1}{\sqrt{2}} & 
\displaystyle
- \frac{1}{\sqrt{2}} & 0 \\
\displaystyle
\frac{1}{\sqrt{3}} - \epsilon' &
\displaystyle
\frac{1}{\sqrt{3}} - \epsilon' &
\displaystyle
\frac{1}{\sqrt{3}} + \epsilon
\end{pmatrix} ,
\end{equation}
where $\epsilon$ and $\epsilon'$ are small and positive and such that $|U_{31}|^2 + |U_{32}|^2 + |U_{33}|^2 = 1$ (the matrix $U$ is then unitary). In order to have $|U_{33}| \geq 1/\sqrt{3}$, the third column should not be permuted. Moreover, $|\det U_{[23][13]}| = |\det U_{[23][23]}| = 1/\sqrt{6} + \epsilon/\sqrt{2}$. Therefore, whatever is the ordering chosen for the first two columns, we have $|\det U_{[23]}| = 1/\sqrt{6} + \epsilon/\sqrt{2}$, which can be made arbitrarily close to $1/\sqrt{6}$. 
\end{proof}

\noindent Using the previous Lemma, we can show the following proposition. 
\begin{proposition*}
Let $M$ be a hierarchical $3\times 3$ matrix. Then it is possible to permute rows and columns in such a way that 
\begin{equation}
|M_{33}| \gtrapprox  \frac{m_3}{3}, \qquad
|\det M_{[23]}| \gtrapprox \frac{m_2m_3}{6} .
\label{eq:orderingM}
\end{equation}
Moreover, it is not possible to set more stringent general bounds. 
\end{proposition*}

\begin{proof} 
It suffices to use the singular value decomposition $M = V^T M_D U$, $U,V$ unitary, $M_D = \diag(m_1,m_2,m_3)$, $0 < m_1 \ll m_2 \ll m_3$. We can then permute the rows and columns of $M$ (i.e.\ the columns of $U$ and $V$) in such a way that $U$ and $V$ satisfy \eq{orderingU}. We then have $|\det M_{33}| \approx |V_{33}\, m_3 U_{33}| \geq m_3/3$ (alternatively, we could have observed that $m^2_3 \approx \Pi^2_1 = \sum_{hk} |M_{hk}|^2$, so that $\max_{hk} |M_{hk}|^2 \gtrapprox m^2_3/9$). Moreover, $|\det M_{[23]}| \approx |\det V_{[33}| |\det U_{]23]}| m_2 m_3 \geq m_2 m_3 /6$. Using the relations above it is also possible to show, as in the Lemma, that the bounds in \eq{orderingM} cannot be made more stringent. 
\end{proof}

\noindent Let us now assume that $M$ satisfies the stability assumption. It is then possible to get stronger bounds on $|M_{33}|$, $|\det M_{[23]}| $. 
\begin{proposition*}
Let $M$ be a hierarchical $3\times 3$ matrix satisfying the stability assumption. Then it is possible to permute rows and columns in such a way that 
\begin{equation}
\label{eq:orderingMbis}
|M_{33}| \gtrapprox  \frac{m_3}{\sqrt{3}}, \qquad
|\det M_{[23]}| \gtrapprox \frac{m_2 m_3}{\sqrt{6}} .
\end{equation}
\end{proposition*}

\begin{proof} 
The result in \eq{orderingMbis} can be proven by direct inspection of the structures allowed by Proposition~\ref{p2}. In particular, we can classify the possible structures in terms of the number $N$ of $2\times 2$ sub-matrices whose determinant is not suppressed with respect to $m_2 m_3$. Note that the stability assumption allows at most $N=4$ such sub-matrices. Indeed, for each sub-determinant giving a unsuppressed contribution, \eq{cof} forces one matrix element to be of  order $m_1$ or smaller, and direct inspection shows that with more than 4 matrix elements of order $m_1$ or smaller, it is not be possible to have 4 or more unsuppressed sub-matrices. Then, direct inspection shows that 
\begin{equation}
\label{eq:boundsN}
\begin{aligned}
|M_{33}| & \gtrapprox  \frac{m_3}{\sqrt{3}}, &  |\det M_{[23]}| \gtrapprox \frac{m_2m_3}{\sqrt{4}} & \qquad \text{($N=4$)} \\
|M_{33}| & \gtrapprox  \frac{m_3}{\sqrt{3}}, &  |\det M_{[23]}| \gtrapprox \frac{m_2m_3}{\sqrt{6}} & \qquad \text{($N=3$)} \\
|M_{33}| & \gtrapprox  \frac{m_3}{\sqrt{3}}, &  |\det M_{[23]}| \gtrapprox \frac{m_2m_3}{\sqrt{2}} & \qquad \text{($N=2$)} \\
|M_{33}| & \gtrapprox  \frac{m_3}{\sqrt{2}}, &  |\det M_{[23]}| \gtrapprox \frac{m_2m_3}{\sqrt{1}} & \qquad \text{($N=1$)} .
\end{aligned}
\end{equation}
We will not go through the lengthy and not particularly inspiring proof, but we make three observations useful to determine the possible structures of $M$ for a given $N$ (and thus to prove \eqs{boundsN}):
\begin{itemize}
\item
The matrix entries must satisfy $|M_{ij}| \leq m_3$, $|M_{ih}M_{jk}| \lesssim m_2 m_3$, $|M_{ih}M_{jk} M_{lm}| \lesssim m_1 m_2 m_3$ when rows and columns are all different. 
\item
The possible structures can be classified by the position of the entries complementary (i.e.\ with no common row or column) to the $2\times 2$ unsuppressed sub-determinants, which by \eq{cof} are not much larger than $m_1$. All remaining $2\times 2$ sub-determinants must be suppressed with respect to $m_2 m_3$. 
\item
Suppose only the $2\times 2$ sub-matrices in the last two rows have unsuppressed determinants and let us consider the two sub-matrices that include the third column elements $M_{23}$ and $M_{33}$, $M_{[23][i3]}$, $i=1,2$. At least one of the two must have $|\det M_{[23][i3]}|\gtrapprox 1/\sqrt{6}$. The latter statement can be shown by observing that if  $|\det M_{[23][i3]}| < \epsilon$ for both $i$, then $(M_{21},M_{22},M_{23}) = (M_{31},M_{32},M_{33}) (M_{23}/M_{33})+ (\delta M_{21},\delta M_{22},0)$, with $|\delta M_{21}|$, $|\delta M_{22}| < \epsilon/ |M_{33}|$ and $|\det M_{[23][21]}| = |\delta M_{21} M_{32} - \delta M_{22} M_{31}| < 2\epsilon$. Therefore, $m^2_2 m^2_3 \approx \Pi^2_2 \approx \sum_{i=1}^3 |\det M_{[23][i3]}|^2 < 6 \epsilon^2$ and $\epsilon \gtrapprox m_2 m_3/\sqrt{6}$. 
\end{itemize}
\end{proof}

\section{Flavor model for $U_\nu = {\bf 1}$}
\label{sec:diagonalneutrinos}

In this appendix we briefly present, as a proof of existence, an abelian flavor model which realises the case in which the neutrino mass matrix is diagonal and the lepton mixing arises from the charged lepton sector, closely related to the one presented in Appendix A of ref.~\cite{Altarelli:2004jb}, albeit with no need of introducing extra messenger fields.
We do so in the context of a supersymmetric $\text{SU}(5)$ grand unified theory.
We introduce a flavor symmetry $F = \U(1)_{F_0} \times \U(1)_{F_1} \times \U(1)_{F_2} \times \U(1)_{F_3} \times \U(1)_{F_4} $.

The relevant field content, as well as charge assignment is given by
\begin{center}
\begin{tabular}{c | cccccccc}
		& $10_1$ & $10_2$& $10_3$ & $\bar{5}_1$ & $\bar{5}_2$ & $\bar{5}_3$ & $5_H$ & $\bar{5}_H$ \\ \hline
	$F_0$ & 3 & 2 & 0 & 0 & 0 & 0 & 0 & 0 \\
	$F_1$ & 2 & 2 & 1 & 1 & 1 & 0 & 0 & 0 \\
	$F_2$ & 2 & 2 & 2 & 1 & 0 & 0 & 0 & 0 \\
	$F_3$ & 2 & 2 & 2 & 0 & 1 & 0 & 0 & 0 \\
	$F_4$ & 2 & 2 & 2 & 0 & 0 & 1 & 0 & 0 \\
\end{tabular}~.
\end{center}
The flavon fields, and their charge assignment, are
\begin{center}
\begin{tabular}{c | cccccccccc}
		& $\theta_0$ & $\theta_1$ & $\theta_2$ & $\theta_3$ & $\theta_4$ & $\theta_5$ & $\theta_6$ & $\theta_7$ & $\theta_8$ & $\theta_9$  \\ \hline
	$F_0$ & -1 & 0 & 0 & 0 & 0 & 0 & 0 & 0 & 0 & 0 \\
	$F_1$ & 0 & -2 & -1 & 0 & 0 & 0 & 0 & 0 & 0 & 0 \\
	$F_2$ & 0 & 0 & 0 & -2 & 0 & 0 & -3 & 0 & 0 & -4 \\
	$F_3$ & 0 & 0 & 0 & 0 & -2 & 0 & 0 & -3 & 0 & -4 \\
	$F_4$ & 0 & 0 & 0 & 0 & 0 & -2 & 0 & 0 & -3 & -4 \\
\end{tabular}~.
\end{center}
The effective superpotential at low energy can be written as
\be
	W = y_{ij} 10_i 10_j 5_H + \eta_{ij} 10_i \bar{5}_j \bar{5}_H + \frac{c_{ij}}{\Lambda} (\bar{5}_i 5_H)(\bar{5}_j 5_H),
\ee
where $\Lambda$ is a high mass scale related to the flavor dynamics and the other couplings are adimensional and include suitable powers of $\langle \theta_i \rangle / \Lambda \sim \lambda \ll 1$ (for simplicity, all vev are assumed to be of the same order) in order to make each term invariant under the symmetry $F$.
This fixes the up-type quark mass matrix to be
\begin{equation}
\label{eq:upquarks}
M_u \sim \hat{y} \langle 5_H \rangle \lambda^2 
\begin{pmatrix}
	\lambda^7 & \lambda^6 & \lambda^4 \\
	\lambda^6 & \lambda^5 & \lambda^3 \\
	\lambda^4 & \lambda^3 & 1
\end{pmatrix} ,
\end{equation}
the charged lepton one to be
\begin{equation}
\label{eq:LeptMassMatr}
M_E \sim \hat{\eta} \langle \bar{5}_H \rangle \lambda^4 
\begin{pmatrix}
	\lambda^4 & \lambda^4 & \lambda^3 \\
	\lambda^3 & \lambda^3 & \lambda^2 \\
	1 & 1 & 1
\end{pmatrix} ,
\end{equation}
and finally the neutrino masses are diagonal and with inverted ordering, proportional to
\be
	M^\nu_{ij} = \delta_{ij} \frac{ \lambda \langle 5_H \rangle^2}{\Lambda} (\delta_{i1} \hat{c}_{1} \lambda + \delta_{i2} \hat{c}_{2} \lambda + \delta_{i3} \hat{c}_{3})~.
\ee
Above we defined the $\mathcal{O}(1)$ parameters $\hat{y}_{ij}$, $\hat{\eta}_{ij}$ and $\hat{c}_{i}$.
Notice that in eq.~\eqref{eq:LeptMassMatr} we reproduced the mass matrix of eq.~\eqref{eq:MEfull2}.
Finally, let us point out that the only symmetries necessary in order to reproduce the texture of eq.~\eqref{eq:LeptMassMatr} in the charged lepton sector (albeit with a different overall scaling with $\lambda$) are the first two $\U(1)$ factors, $\U(1)_{F_0} \times \U(1)_{F_1}$, and the only flavons necessary are $\theta_0, \theta_1$ and $\theta_2$, with the same charges as specified above.

\end{document}